\def\rank{\mathop{\rm rank}\nolimits}
\def\diag{\mathop{\rm diag}}
\def\wgt{\mathop{\rm wgt}\nolimits}

\def\ker{\mathop{\rm Ker}}
\def\im{\mathop{\rm im}}
\def\tr{\mathop{\rm tr}\nolimits}

\def\css{\mathop{\rm CSS}\nolimits}
\def\F{\mathbb{F}}

\documentclass[aps,pra,
twocolumn,
notitlepage,longbibliography]{revtex4-2}
\usepackage{hyperref}
\usepackage{braket}
\usepackage{xcolor}
\usepackage{amsmath,amsfonts,amssymb}
\usepackage{enumitem}

\usepackage{amsthm}
\newtheorem{theorem}{Theorem}

\newtheorem{statement}[theorem]{Statement}

\newtheorem{lemma}[theorem]{Lemma}

\newtheorem{conjecture}[theorem]{Conjecture}
\newtheorem{example}[theorem]{Example}
\begin{document}
\title{Minimal distances for certain quantum product codes\\ and tensor products of chain complexes}
\date\today
\author{Weilei Zeng}
\address{Department of Physics \& Astronomy, University of California,
  Riverside, California 92521, USA}
\author{Leonid P.\ Pryadko}
\address{Department of Physics \& Astronomy, University of California,
  Riverside, California 92521, USA}
\email{leonid.pryadko@ucr.edu}

\begin{abstract}
  We use a map to quantum error-correcting codes and a subspace
  projection to get lower bounds for minimal homological distances in
  a tensor product of two chain complexes of vector spaces over a
  finite field.  Homology groups of such a complex are described by
  the K\"unneth theorem.  We give an explicit expression for the
  distances when one of the complexes is a linear map between two
  spaces.  The codes in the construction, subsystem product codes and
  their gauge-fixed variants, generalize several known families of
  quantum error-correcting codes.
\end{abstract}

\maketitle


\section{Introduction}

Central idealization in topology is the focus on continuity while
sizes are ignored.  Topologically speaking, an opening through a straw
is no different than a pinhole in a piece of paper, or a missing pixel
in an image.  Yet a missing pixel could be just an artifact of the
noisy data.  No wonder that in practical applications the sizes and
distances \emph{are} important, and are incorporated into
computational algorithms in a variety of
ways\cite{Kaczynski-Mischaikow-Mrozek-2003,%
  Kaczynski-Mischaikow-Mrozek-book-2004,%
  GonzalezLorenzo-Bac-Mari-Real-2016,Audoux-Couvreur-2017,Dey-Li-Wang-2018}.

Quantum stabilizer and, more generally, subsystem codes offer an
excellent example of a problem where such a distance is extremely
relevant\cite{Bombin-MartinDelgado-2007,%
  Bullock-Brennen-2007,Audoux-Couvreur-2017}.  Namely, a qubit quantum
stabilizer code is isomorphic to a chain complex ${\cal C}$ with three
finite-dimensional binary spaces, where logical operators correspond
to elements of the first homology group $H_1({\cal C})$.  In the case
of a Calderbank-Shor-Steane\cite{Calderbank-Shor-1996,Steane-1996}
(CSS) code, the rank of this group gives the number $k$ of encoded
qubits, the code length $n$ is the dimension of the corresponding
space ${\cal C}_1$, while the distance $d$ of the quantum error
correcting code (QECC), the minimum weight of a non-trivial element in
$H_1({\cal C})$ (or the corresponding co-homology group), has to be
sufficiently large for the code to offer a protection against
environmental errors.

In fact, topological QECCs, generalizations of the toric
code\cite{Bravyi-Kitaev-1998,Freedman-Meyer-1998,%
  Dennis-Kitaev-Landahl-Preskill-2002,Bombin-MartinDelgado-2007,%
  Castelnovo-Chamon-2008,Mazac-Hamma-2012,%
  Bombin-Chhajlany-Horodecki-MartinDelgado-2013} invented by
Kitaev\cite{kitaev-anyons}, are presently at the crux of research in
quantum error correction (QEC).  Such a code can be constructed from
any tessellation of an arbitrary surface or a higher-dimensional
manifold.  The essential advantage of topological codes is locality:
stabilizer generators, operators to be measured frequently, involve
only qubits in the immediate vicinity of each other; this is what
makes planar surface codes so attractive and practical.  However,
locality also limits the parameters of topological
codes\cite{Bravyi-Terhal-2009,Bravyi-Poulin-Terhal-2010,%
  Delfosse-2013,Flammia-Haah-Kastoryano-Kim-2017}.  In particular, for
a code of length $n$ with generators local in two dimensions, the
number of encoded qubits $k$ and the minimal distance $d$ satisfy the
inequality\cite{Bravyi-Terhal-2009} $kd^2\le {\cal O}(n)$.  This
implies asymptotically zero rate, $R=k/n\to0$, whenever $d$ diverges
with $n$.

More general quantum low-density parity-check (LDPC) codes have
stabilizer generators of bounded weight but no locality constraint.
This is the only class of codes known so far to combine finite rates
with non-zero fault-tolerant (FT)
thresholds\cite{Kovalev-Pryadko-FT-2013,Dumer-Kovalev-Pryadko-bnd-2015},
to allow scalable quantum computation with a finite multiplicative
overhead\cite{Gottesman-overhead-2014}.  However, unlike in the
classical case, where capacity-approaching codes can be constructed
from random sparse matrices\cite{Gallager-1962,Gallager-book-1963,%
  Litsyn-Shevelev-2002,%
  Richardson-Shokrollahi-Amin-Urbanke-2001}, matrices suitable for
constructing quantum LDPC codes are highly atypical in the
corresponding ensembles.  Thus, an algebraic ansatz is required to
construct large-distance quantum LDPC codes.  Precious few examples
of algebraic constructions are known that give finite rate codes and
also satisfy sufficient conditions\cite{Dumer-Kovalev-Pryadko-bnd-2015} for
fault-tolerance: bounded weight of stabilizer generators and minimum
distance that scales logarithmically or faster with the block length
$n$.  Such constructions include hyperbolic codes on two- and
higher-dimensional
manifolds\cite{Zemor-2009,Delfosse-Zemor-2010,Breuckmann-Terhal-2015,
  Breuckmann-Vuillot-Campbell-Krishna-Terhal-2017,Guth-Lubotzky-2014},
and quantum hypergraph-product (QHP) \& related
codes\cite{Tillich-Zemor-2009,Tillich-Zemor-2014,%
  Kovalev-Pryadko-Hyperbicycle-2013,Zeng-Pryadko-2018}.  Further, some
constructions, e.g., in Refs.\
\onlinecite{Couvreur-Delfosse-Zemor-2012,Bravyi-Hastings-2013,%
  Audoux-2014,%
  Audoux-Couvreur-2017,Hastings-codes-2016}, have finite rates and
relatively high distances, with the stabilizer generator weights that
grow with $n$ logarithmically.  It is not known whether these codes
have non-zero FT thresholds.  However, such codes can be modified into
those with provable FT thresholds with the help of weight
reduction\cite{Hastings-weight-2016,Hastings-Haah-ODonnell-2020}.

The original QHP ansatz\cite{Tillich-Zemor-2009} by Tillich and
Z\'emor can be seen as a tensor product of two chain complexes
${\cal A}$ and ${\cal B}$, each involving just two finite-dimensional
binary spaces with chosen bases, so that the corresponding boundary
operators are just binary matrices without any additional constraints.
The resulting chain complex has three spaces; elements of the first
homology group of dimension $k=\rank H_1({\cal A}\times {\cal B})$
form a half of the quantum code
$\mathcal{Q}_1(\mathcal{A}\times \mathcal{B})$ encoding $k$ qubits
(the other half comes from the corresponding co-homology group).  This
dimension can be immediately recovered from the K\"unneth
formula\cite{Kunneth-1923,Kunneth-1924}.  The main result by Tillich
and Z\'emor is the expression for the minimal distance.  This was
generalized by the present authors to homology groups in a tensor
product of a general chain complex over binary spaces with that
involving just two spaces\cite{Zeng-Pryadko-2018}.

In this work, we offer a generalization of the distance result in
Ref.~\onlinecite{Zeng-Pryadko-2018} to a tensor product of two chain
complexes of vector spaces over any finite field $F$, with one of the
complexes still required to be a linear map between a pair of spaces.
While the original proof\cite{Zeng-Pryadko-2018} would still work with
a general field, here we give a simpler proof for the lower bound on
the minimum distance, formulated in terms of a \emph{projected}
product complex with the level-$j$ subspace projected onto just one
subspace
${\cal A}_i\otimes {\cal B}_{j-i}\subset ({\cal A}\times{\cal B})_j$.
As a result of the projection, the quantum code
$\mathcal{Q}_j(\mathcal{A}\times \mathcal{B})$ associated with the
$j$-th homology group of the product complex is replaced by an
$F$-linear quantum subsystem code; its distance gives a lower bound on
the distance associated with the homology group
$H_j(\mathcal{A}\times \mathcal{B})$ of the original product
complex.  
When one of the complexes has length two, the minimum distance of the
subsystem code can be computed and, as in the binary case, the result
saturates the upper bound.

While the construction also works for a product of chain complexes of
arbitrary length, we failed to find a tight lower bound on the
distance of the corresponding projected codes.  Further, we have found
a class of examples, a generalization of the homological product of
Steane code with
itself\cite{Bravyi-Hastings-2013,Audoux-Couvreur-2017}, where the
distance in the projected complex is strictly smaller than the upper
bound.  However, through extensive numerics for
$q\in\{2,3,2^2,5,7,2^3,3^2,11\}$, we could not find a single case
where the homological distance in the full product complex would fail
to saturate the upper bound.  We conjecture that in a product of
general chain complexes, the upper bound on the homological distance
is saturated.

\textsc{Potential applications:} In theory of QEC, in addition to
defining new classes of quantum LDPC codes with parameters known
explicitly, our construction may be useful for (i) optimizing repeated
measurements in the problem of FT quantum error
correction\cite{Wang-Harrington-Preskill-2003,
Andrist-PhD-2012,
  Kovalev-Pryadko-FT-2013,Dumer-Kovalev-Pryadko-bnd-2015},
   (ii) related problem of single-shot error
correction\cite{Fujiwara-2014,Ashikhmin-Lai-Brun-2014,%
  Ashikhmin-Lai-Brun-2016,Campbell-2018},
(iii) analysis of transformations between different QECCs, like the
distance-balancing trick by Hastings\cite{Hastings-weight-2016}, and
(iv) construction of asymmetric quantum CSS codes optimized for operation
where error rates for $X$ and $Z$ channels may differ
strongly\cite{Ioffe-Mezard-2007,Evans-2007,Stephens-2008,%
  Aliferis-Preskill-2008,sarvepalli-2009,Tuckett-Bartlett-Flammia-2018}.

More generally, K\"unneth formula is one of the most important and
widely known results in algebraic topology, see, e.g.,
Ref.~\onlinecite{Weibel-book-1994}.  Its well known consequence is the
relation between the Betti numbers of two manifolds and their product,
which can be written in terms of a product of the corresponding
generating functions, the Poincare polynomials
$p(x)=\mathrm{b}_0+\mathrm{b}_1x+\mathrm{b}_2x^2+\ldots$.  Generally,
$\mathrm{b}_k$ is the rank of the $k$\,th homology group.  For
manifolds in three dimensions, the zeroth Betti number,
$\mathrm{b}_0$, gives the number of connected components, the first,
$\mathrm{b}_1$, the number of one-dimensional holes (incontractible
cycles), and $\mathrm{b}_2$ the number of closed surfaces that cut out
internal cavities.  In particular, for a torus, $p(x)=1+2x+x^2$, which
can be written as $(1+x)^2$, the square of the Poincare polynomial for
a circle.

Our results can be seen as equipping K\"unneth formula with a
distance.  For example, consider a torus defined via periodic boundary
conditions on a plane, e.g., with periods $L_x$ and $L_y$ along the
$x$ and $y$ directions. Then, the systola (girth in the case of a
graph) is $\min(L_x,L_y)$, while the surface area (number of
plaquettes) is $L_xL_y$.  More generally, for a tensor product of a
circle with perimeter $L$ and an arbitrary manifold with
systola $ L_1'$, minimum surface area $L_2'$, etc., the corresponding
dimensions are given by $\min(L,L_1')$, $\min(LL_1',L_2')$, \ldots

The outline of the rest of the paper is as follows.  In
Sec.~\ref{sec:background} we go over the necessary background facts
from theory of classical and quantum error-correcting codes, as well
as chain complexes of vector spaces over a finite field $F$.  We also
establish the relation between (co)homology groups in such a complex
and $F$-linear quantum codes.  In Sec.~\ref{sec:product-codes} we
describe the construction and derive upper and lower bounds for
minimal distances of several related families of ``product'' codes
constructed in terms of Kronecker products of matrices associated with
a pair of quantum codes whose parameters are known.  In
Sec.~\ref{sec:complexes} we formulate main results in application to
chain complexes, give detailed proofs, and discuss their use in
fault-tolerant quantum error correction.  Finally, in
Sec.~\ref{sec:extensions}, we discuss some extensions of present
results.

\section{Background}
\label{sec:background}

\subsection{Classical $q$-ary codes}
A classical $q$-ary code\cite{MS-book} ${\cal C}$ with parameters
$(n,K,d)_q$ is a collection of $K$ strings (codewords) of length $n$
over an alphabet with $q$ symbols.  The code distance $d$ is the
minimum number of positions where two strings in the code differ.  A
linear $q$-ary code, where $q$ is a power of a prime, is a
$k$-dimensional subspace of the $n$-dimensional vector space $F^n$
over the field $F\equiv \mathbb{F}_q$.  Such a code contains $K=q^k$
strings.  A linear code ${\cal C}\equiv{\cal C}_G$ with parameters
$[n,k,d]_q$ can be defined in terms of a generator matrix $G$ whose
rows are the chosen basis vectors; the dimension $k$ of the code
$\mathcal{C}_G$ is $k=\rank G$.  For a linear code, the distance $d$
is the minimum Hamming weight of a non-zero vector in the code.

A linear subspace in $F^n$ can be also specified in terms of its
orthogonal subspace.  To this end, one has to choose the inner product
to be used\cite{Rains-Sloane-self-dual-codes-1998,%
  Ketkar-Klappenecker-Kumar-Sarvepalli-2006,Fan-Zgang-2017}.  The
simplest choice is the usual Euclidean scalar product,
$a\cdot b\equiv a\,b^T$, where $a, b\in {F}^n$ are considered as
length-$n$ row vectors, and $b^T$ is the transposed vector.
Respectively, the dual $\mathcal{C}^\perp$ of a linear code
$\mathcal{C}$ is a collection of $q$-ary row vectors orthogonal to any
vector in $\mathcal{C}$,
\begin{equation}
  {\cal C}^\perp =\{ b\in F^n|  c\, b^T=0, \forall c\in\mathcal{C}\}.
  \label{eq:dual-code}
\end{equation}
For a linear code of size $|\mathcal{C}|=q^k$ and dimension $k$, the
dual code has size $|\mathcal{C}^\perp|=q^{n-k}$.  Generator matrix
$H$ of the dual code, ${\cal C}_H\equiv \mathcal{C}_G^\perp$, is called a
parity check matrix of the original code.  More generally, a pair of
$n$-column matrices $G$ and $H$ with elements in $F$ are called
\emph{mutually dual} if
\begin{equation}
GH^T=0,\quad \rank G+\rank H=n. \label{eq:dual-matrix}
\end{equation}

Given a string $c\in{F}^n$, denote $V\equiv \{1,2,\ldots,n\}$ the set
indexing the individual characters.  For any \emph{index set}
$I\subseteq V$ of length $|I|=r$, let $c[I]\in{F}^r$ be a substring of
$c$ with the characters in all positions $i\not\in I$ dropped.
Similarly, for an $n$-column matrix $G$ with rows $g_j$, $G[I]$ is
formed by the rows $g_j[I]$.  If ${\cal C}={\cal C}_G$ is an
$F$-linear code with the generating matrix $G$, then the code of length
$|I|$ with the generating matrix $G[I]$ is the code \emph{punctured}
outside $I$, ${\cal C}_{\rm p}[I]\equiv\{c[I]\,|\, c\in{\cal C}\}$.

The \emph{shortened} code ${\cal C}_{\rm s}[I]$ is formed similarly,
except only from the codewords supported inside $I$,
${\cal C}_{\rm s}[I]=\{c[I]\,|\,c=(c_1,c_2,\ldots,c_n)\in {\cal C}$ and
$c_i=0$ for each $i\not\in I\}$.  The dual of a punctured code
$\mathcal{C}_{\rm p}[I]$ is the shortened dual code,
$(\mathcal{C}_{\rm p}[I])^\perp=(\mathcal{C}^\perp)_{\rm s}[I]$.  To
express this relation in terms of matrices, consider a pair of
mutually dual matrices in Eq.~(\ref{eq:dual-matrix}) and a code
$\mathcal{C}\equiv \mathcal{C}_G=\mathcal{C}_H^\perp$.  Denote a
generator matrix of the shortened code $\mathcal{C}_{\rm s}[I]$ as
$G_I$.  Duality between the punctured original and the shortened dual
codes implies that the corresponding generator matrices $G_I$ and
$H[I]$ are also mutually dual\cite{MS-book},
\begin{equation}
  \label{eq:dual-shortening}
  H[I]\, G_I^T=0,\quad \rank G_I+\rank H[I]=|I|.
\end{equation}
Similarly, $H_I$ is a dual of the punctured matrix $G[I]$.

In relation to quantum codes, we  also consider $q$-ary linear
space $F^{2n}$ of length-$2n$ vectors in the form $e=(a|b)$, where
both $a$ and $b$ are row vectors of length $n$.  The \emph{symplectic
  product} of two such vectors is defined as
\begin{equation}
  \label{eq:symplectic-product}
  e'\star e\equiv a'\cdot b-b'\cdot a\equiv e'\Sigma\, e^T.
\end{equation}
The right-most form contains the symplectic matrix,
\begin{equation}
  \label{eq:symplectic-matrix}
  \Sigma\equiv \Sigma_n=\left(
    \begin{array}[c]{cc}
      &I_n\\ -I_n& 
    \end{array}\right),
\end{equation}
with $I_n$ an $n\times n$ identity matrix.  For a row vector
$e\in\mathbb{F}^{2n}$, the (symplectic) conjugate is
$\tilde{e}= e\,\Sigma^T=-e\,\Sigma$, so that the symplectic product
can be also written as $e'\star e={e}'\,\tilde{e}^T$.  The code
orthogonal with respect to the symplectic product to a given $q$-ary
code $\mathcal{C}\subseteq F^{2n}$ is denoted
$\mathcal{C}^{\perp\star}$.  A code $\mathcal{C}_G^{\perp\star}$
orthogonal to $\mathcal{C}_G$ has generator
matrix $ G^\star$, a (symplectic) parity check matrix of
the original code $\mathcal{C}_G$ and also a Euclidean dual of the
matrix $\tilde{G}=-G\Sigma$, see Eq.~(\ref{eq:dual-matrix}), except
that the code length here is $2n$.  Explicitly, for a generator matrix
in the block form $G=(A|B)$, where each block has $n$ columns, rows of
$G^\star$ are orthogonal to the rows of $\tilde{G}=(B|-\!A)$,
$\tilde{G}(G^\star)^T=0$.

\subsection{Quantum stabilizer codes over qudits}
\label{sec:background-qecc}

A single qudit is an isolated quantum-mechanical system whose pure
states are described by vectors $\ket\psi$ in a $q$-dimensional
Hilbert space ${\cal H}_q$.  Pure states of $n$ qudits are described by
vectors in the Hilbert space ${\cal H}_q^{\otimes n}$, the tensor
product of $n$ single-qudit spaces.  The corresponding physical
observables are described by Hermitian operators acting in
${\cal H}_q^{\otimes n}$.  An $n$-qudit quantum error-correcting code
${\cal Q}$ with parameters $((n,K))_q$ is a $K$-dimensional subspace
of ${\cal H}_q^{\otimes n}$.

When $q=p^m$ is a power of a prime, there is a particularly nice basis
for single-qudit operators acting in ${\cal H}_q$.  Following
Ref.~\onlinecite{Ketkar-Klappenecker-Kumar-Sarvepalli-2006},  choose
$q$ orthonormal basis vectors $\ket{z}\in{\cal H}_q$, $z\in F$, enumerated
by elements of the finite field $F\equiv \F_q$.  Two kinds of unitary
operators, $\hat X(a)$ and $\hat Z(a)$, $a\in F$, also enumerated by
elements of the field, are defined in terms of their action on the
basis vectors,
\begin{equation}
  \hat X(a)\ket{z}=\ket{z+a},\quad \hat Z(b) \ket{z} = \omega^{ \tr(bz) }\ket{z},
  \label{eq:X-Z-operators}
\end{equation}
where, with $q=p^m$ a prime power,
\begin{equation}
  \tr(x) \equiv \tr_{F/\mathbb{F}_p}(x)= x+x^p+\ldots +x^{p^{m-1}}
  \label{eq:field-trace}
\end{equation}
is the trace operation from the extension field $F=\mathbb{F}_q$ to
the prime field $\mathbb{F}_p$, and $\omega = e^{ 2 \pi i/p}$ is a
primitive $p$\,th root of unity.  The basis of interest is formed by
the $q^2$ operators $\hat X(a)\hat Z(b)$, $a,b\in F$.

The same operators can be used to construct a basis of operators
acting in an $n$-qudit Hilbert space $\mathcal{H}_q^{\otimes n}$.
Namely, given a $q$-ary vector $a\in F^n$, define the $n$-qudit
operators $\hat X(a)$ and $\hat Z(a)$ as tensor products over
components, e.g.,
$\hat X(a)=\hat X(a_1) \otimes \hat X(a_2)\otimes \ldots \otimes \hat
X(a_n)$.  These operators generate the $n$-qudit Pauli group%
\begin{equation}
  \mathcal{P}_n=\left\{ \omega^c \hat X(a)\hat Z(b)|c\in \mathbb{F}_p,\,
    a,b\in F^n\right\}. 
  \label{eq:Pauli-group}
\end{equation}
The weight $\wgt(\hat{U})$ of
an operator $\hat U\in \mathcal{P}_n$ is defined as the number of
qudits that $\hat U$ acts upon non-trivially.
Up to a phase, a Pauli operator
$\hat U(a,b;c)\equiv \omega^c\hat X(a)\hat Z(b)$ can be specified by
the vector $e\equiv (a|b)\in \mathbb{F}_q^{2n}$.  The commutation
relation between two such operators (with inessential phase factors
suppressed) reads
\begin{equation}
  \label{eq:commutation}
  \hat{U}(a,b)\hat{U}(a',b')  =\omega^{\tr(a \cdot b'-b\cdot a')} \hat{U}(a',b') \hat{U}(a,b). 
\end{equation}
In particular, the two operators commute if and only if the trace
symplectic form $ \tr(a \cdot b'- b\cdot a')$ vanishes.

An $n$-qudit
stabilizer code is a common $+1$ eigenspace of all operators in a
\emph{stabilizer} group $\mathcal{S}$,%
\begin{equation}
  \label{eq:stabilizer-code}
  \mathcal{Q}\equiv \mathcal{Q}_{\cal S}=\left\{\ket\psi\in {\cal H}_q^{\otimes n}\,\Bigl\vert\,
    \hat U\ket\psi=\ket\psi,\forall \hat U\in\mathcal{S}\right\},
\end{equation}
where $\mathcal{S}$ is an abelian subgroup of $\mathcal{P}_n$ whose
only zero-weight member is the identity operator.  It is easy to see
that any Pauli operator $\hat E$ which does not commute with an
element of the stabilizer throws the code $\mathcal{Q}_\mathcal{S}$
into an orthogonal space $\hat{E} \mathcal{Q}_\mathcal{S}$; such
operators are called \emph{detectable} errors.  \emph{Undetectable}
errors commute with all elements of $\mathcal{S}$.  In particular, all
elements of $\mathcal{S}$ are undetectable.  However, since these
operators act trivially in the code, such errors can be ignored.  Only
undetectable errors outside of $\mathcal{S}$ (up to a phase) are
relevant for error correction.  Such errors act non-trivially in the
code and correspond to \emph{logical} operators.  The \emph{distance}
$d$ of a stabilizer code is defined as the minimum weight of an
undetectable Pauli operator not equal (up to a phase) to an element of
$\mathcal{S}$.  Similarly, errors $\hat{E}\in\mathcal{P}_n$ and
$\hat{E}'=\omega^c\hat{S} \hat{E}$ that differ by an element
$\hat{S}\in\mathcal{S}$ of the stabilizer group (again, up to a phase)
are called mutually \emph{degenerate}; for all practical purposes such
errors are equivalent.

Up to the choice of the phases of its generators, a stabilizer group
can be also represented as a length-$2n$ additive code over
$\mathbb{F}_q$, isomorphic to a length-$2nm$ linear code over the
prime field $\mathbb{F}_p$, where $q=p^m$.  The commutation condition
gives an additional requirement that the rows of the generator matrix
be mutually orthogonal with respect to the symplectic trace product.
In general, any element $x\in\mathbb{F}_q$ of an extension of a field
of prime degree $p$ is $p$-periodic with respect to addition,
$p\,x=0$.  Respectively, the size of a stabilizer group is a power of
the prime $p$.  This gives the code dimension
$K=q^n/|\mathcal{S}|=p^s$, which is not necessarily an integral power
of $q$.  Thus, excluding the case of a prime field analyzed in
Ref.~\onlinecite{Ashikhmin-Knill-2001}, a stabilizer code does not
necessarily encodes an integer number of qudits.  The latter condition
is satisfied under an additional constraint, $s\bmod m=0$.

\subsection{$F$-linear quantum codes}

In this work we focus on the special case of $F$-linear length-$2n$
codes formed by vectors of the form $e=(a|b)$, $a,b\in F^n$, and
duality implemented in terms of the Euclidean symplectic product (\ref{eq:symplectic-product}).  Unlike in Eq.~(\ref{eq:commutation}),
there is no field trace in this expression.  Thus, $e\star e'=0$ gives
a sufficient but not a necessary condition for the Pauli operators
$\hat{U}(e)$ and $\hat{U}(e')$ to commute, unless $q$ is a prime.
Such an approach follows the definition of CSS codes in
Ref.~\onlinecite{Ketkar-Klappenecker-Kumar-Sarvepalli-2006}.
Alternatively, many of the same results can be obtained by classifying
generators in terms of a lifted Pauli group as suggested by
Gottesman\cite{Gottesman-prime-power-2014}.

Degeneracy is the key difference of quantum codes from their classical
counterparts.  Two vectors $e$ and $e'$ in $F^{2n}$ are called
\emph{degenerate} with respect to elements of the $F$-linear code
$\mathcal{C}_G$ generated by an $r\times 2n$ matrix $G$ iff there
exists an $\alpha\in F^r$ such that $e'=e+ \alpha G$.  Degeneracy with
respect to $\mathcal{C}_G$ is denoted $e'\stackrel{G}\simeq e$, where
the generating matrix may be omitted if the meaning is clear from
context.

In the simplest case rows of the generator matrix $H=(A|B)$ (here and
below denoted as $H$ to indicate that orthogonality is expected) are
mutually orthogonal with respect to the symplectic product,%
\begin{equation}
H\tilde{H}^T\equiv H\Sigma H^T=AB^T-BA^T=0,\label{eq:symplectic-orthogonality}
\end{equation}
which is equivalent to
$\mathcal{C}_H\subseteq \mathcal{C}_H^{\perp\star}$.  The space
$\mathcal{C}_H$ is readily seen as the symplectic map of a stabilizer
group acting in $\mathcal{H}_q^{\otimes n}$.  The corresponding dual
code $\mathcal{C}_H^{\perp\star}$, with any pair of vectors degenerate
with respect to $\mathcal{C}_H$ identified, is called an $F$-linear
\emph{stabilizer code}.  The same object is also known as the quotient
 space $\mathcal{C}_H^{\perp\star}/\mathcal{C}_H$.

 Given any set of $(m\rank H)$ additively independent basis vectors of
 $\mathcal{C}_H$, a stabilizer group
 $\mathcal{S}\subseteq \mathcal{P}_n$ can be constructed by assigning
 each generator a phase $c\in \mathbb{F}_p$.  With this map, vectors
 in $\mathcal{C}_H^{\perp\star}$ correspond (up to a phase) to
 undetectable Pauli errors, i.e., operators acting in the space
 $\mathcal{Q}_\mathcal{S}\subseteq \mathcal{H}_q^{\otimes n}$
 stabilized by $\mathcal{S}$.  Stabilizer group being abelian, it is a
 subgroup of the group $\mathcal{L}_\mathcal{S}$ of all undetectable
 Pauli errors acting in $\mathcal{H}_q^{\otimes n}$.  Thus, mutually
 non-degenerate logical operators are classified by elements of the
 quotient group $\mathcal{L}_\mathcal{S}/\mathcal{S}$.  If we ignore
 the phases, then this group is isomorphic to the $F$-linear
 stabilizer code $\mathcal{C}_H^{\perp\star}/\mathcal{C}_H$.  Notice
 that the subspace
 $\mathcal{Q}_\mathcal{S} \subseteq \mathcal{H}_q^{\otimes n}$ is also
 called a stabilizer code, but this should not cause a confusion as we
 will exclusively use the former meaning.

For an $F$-linear stabilizer code based on the generator matrix $H$,
any codeword $c$ satisfies $\tilde{H}c^T=0$, see
Eq.~(\ref{eq:symplectic-orthogonality}); equivalent codewords are
mutually degenerate, $c'\stackrel{H}\simeq c$.  Using
orthogonalization, we can construct $k=n-\rank H$ pairs of canonically
conjugated codewords $c_i$, $c_i'$ such that
$c_i\star c_j'=\delta_{ij}$, $i,j\le k$.  Equivalently, we can
construct a \emph{logical generator} matrix $L$ whose rows are
orthogonal to those of $\tilde{H}$, $\tilde{H}L^T=0$, are linearly
independent from rows of $H$, and, in addition,
 \begin{equation}
   L\Sigma_n L^T=\Sigma_k.\label{eq:logical-matrix}
\end{equation}

More generally, with $\tilde{G}G^T$ not necessarily zero,
$\mathcal{C}_G$-degeneracy classes of different vectors in
$\mathcal{C}_G^{\perp\star}$ correspond to an $F$-linear
\emph{subsystem code}, a generalization of qubit subsystem
codes\cite{Poulin-subs-2005,Bacon-subs-2006}.  Elements of
$\mathcal{C}_G$ form a symplectic map of subsystem code's \emph{gauge
  group}, while vectors $c\in\mathcal{C}_G^{\perp\star}$ correspond to
\emph{bare} logical operators.  Multiplication of a bare logical
operator $\hat U(c)$ by an element of the gauge group gives a
\emph{dressed} logical operator; with the symplectic map this
corresponds to adding a linear combination of the rows of $G$.
Nonequivalent logical operators in $\mathcal{P}_n$ map to vectors in
$F^{2n}$ which are not degenerate with respect to ${\cal C}_G$,
$c'{\not\simeq} c$.

A subsystem code can also be defined in terms of a stabilizer code
whose stabilizer group maps to the space
$\mathcal{C}_{H}\equiv \mathcal{C}_G\cap \mathcal{C}_{G}^{\perp\star}$
of dimension $r=\rank G-2\kappa$, where
$2\kappa=\rank (G\tilde{G}^T)$.  The space $\mathcal{C}_H$ is
generated by code's \emph{stabilizer generator} matrix $H$ whose rows
are linear combinations of the rows of $G$, and also $\tilde{G}H^T=0$.
The corresponding orthogonal space $\mathcal{C}_H^{\perp\star}$
contains $k+\kappa=n-r$ canonically conjugated vector pairs, including
$\kappa$ such pairs in $\mathcal{C}_G$ (these correspond to gauge
qudits) and $k$ pairs in
$\mathcal{C}_G^{\perp\star}\setminus \mathcal{C}_G$ corresponding to
logical operators of the data qudits.

In the following, we will be mostly interested in CSS
codes\cite{Ketkar-Klappenecker-Kumar-Sarvepalli-2006}, a special class
of $F$-linear subsystem (or stabilizer) codes whose generator matrices
can be chosen in a block-diagonal form, $G=\diag(G_X,G_Z)$, with each
block containing $n$ columns.  The corresponding stabilizer generator
matrix also has a block form, $H=\diag(H_X,H_Z)$; the symplectic
orthogonality is equivalent to $G_XH_Z^T=0$ and $G_ZH_X^T=0$.  Such a
code, denoted $\css(G_X,G_Z)$, is a direct sum of an $X$- and a
$Z$-like codes,
\begin{equation}
  \css(G_X,G_Z)=\mathcal{C}_X\oplus\mathcal{C}_Z=
  \mathcal{C}_{H_Z}^\perp/\mathcal{C}_{G_X}\oplus
  \mathcal{C}_{H_X}^\perp/\mathcal{C}_{G_Z},   
  \label{eq:css-code-decomp}
\end{equation}
where each term in the right-hand side (r.h.s.) is a quotient of two linear spaces.
Clearly, the spaces $\mathcal{C}_X$ and $\mathcal{C}_Z$ are identical
to those in \emph{gauge-fixed} stabilizer codes with generator
matrices $H_1=\diag (G_X,H_Z)$ and $H_2= \diag (H_X,G_Z)$,
respectively.  Gauge generator matrix contains $\kappa$ conjugate
vector pairs not in $\mathcal{C}_H$, thus $\rank G_X=\rank H_X+\kappa$
and $\rank G_Z=\rank H_Z+\kappa$.  As a result, both codes in the
r.h.s.\ of Eq.~(\ref{eq:css-code-decomp}) contain
$k=n-\rank H_X-\rank G_Z$ inequivalent vectors.  The distances of the
two codes are
\begin{equation}
d_X=\min_{x\in \mathcal{C}_{H_Z}^\perp \setminus \mathcal{C}_{G_X}}
\wgt(x), \;
d_Z=\min_{x\in \mathcal{C}_{H_X}^\perp \setminus \mathcal{C}_{G_Z}}
\wgt(x).\label{eq:dX}
\end{equation}
Any $k$ inequivalent codewords from $\mathcal{C}_X$ can be chosen to
form the rows of a logical generator matrix $L_X$; in general
$L_XH_Z^T=0$.  However, it is convenient to choose bare codewords for
the basis, so that also $L_X G_Z^T=0$.  Using bare codewords for the
basis of the logical generator matrix of the other code, $L_Z$, this
matrix will satisfy $L_ZG_X^T=0$.  In addition, choosing conjugate
vector pairs for the two bases, we can also ensure
\begin{equation}
L_XL_Z^T=I_k;\label{eq:cw-orthogonality}
\end{equation}
with the full-code logical generator matrix in the block-diagonal
form, $L=\diag(L_X,L_Z)$. This is the CSS form of Eq.\
(\ref{eq:logical-matrix}).  Parameters of such a CSS code are denoted
as $[[n,k,(d_X,d_Z)]]_q$, where the usual code distance is given by the minimum, 
$d=\min(d_X,d_Z)$.

\subsection{Chain complex of $F$-linear spaces.}

Generally, a \emph{chain complex} is a sequence of abelian groups and
a sequence of homomorphisms (boundary operators) between pairs of
consecutive groups such that the image of each homomorphism be
included in the kernel of the next.  Here we will be concerned with
the special case of chain complexes of finite-dimensional vector
spaces $\ldots,\mathcal{A}_{j-1}, \mathcal{A}_j,\ldots$ over a finite
field $F=\mathbb{F}_q$, where $q=p^m$ is a power of a prime $p$.  In
this case the boundary operators are linear transformations
$\partial_j:{\cal A}_{j-1}\leftarrow {\cal A}_j $ that map between
each pair of neighboring spaces, with the requirement
$\partial_j\partial_{j+1}=0$, $j\in\mathbb{Z}$.  We define an
$\ell$-complex ${\cal A}\equiv {\cal K}(A_1,\ldots,A_\ell)$, a bounded
chain complex which only contains $\ell+1$ non-trivial spaces with
fixed bases, in terms of $n_{j-1}\times n_j$ matrices $A_j$ with
elements from $F$ serving as the boundary operators,
$j\in\{1,\ldots,\ell\}$:
\begin{equation}
  \label{eq:chain-complex}
  {\cal A}:\;
  \ldots \leftarrow\{0\}\stackrel{\partial_0}\leftarrow {\cal
    A}_0\stackrel{A_1}\leftarrow  
  {\cal A}_1\ldots \stackrel{A_{\ell}}\leftarrow
  {\cal A}_{\ell}\stackrel{\partial_{\ell+1}}\leftarrow
  \{0\} \ldots
\end{equation}
Here the neighboring matrices must be mutually orthogonal,
$A_{j-1}A_{j}=0$, $j\in\{2,\ldots,\ell\}$.  In addition to boundary
operators given by the matrices $A_j$, implicit are the trivial
operators $\partial_0:\{0\}\leftarrow {\cal A}_0$
and $\partial_{\ell+1}: {\cal A}_\ell\leftarrow \{0\}$ (with the image
being the zero vector in ${\cal A}_\ell$)
treated formally as rank-zero $0\times n_0$ and $n_\ell\times 0$
matrices.

Elements of the subspace $\im (\partial_{j+1})\subseteq {\cal A}_j$
are called boundaries; in our case these are linear combinations of
columns of $A_{j+1}$ and, therefore, form a binary linear code with
the generator matrix $A_{j+1}^T$,
$\im (A_{j+1})=\mathcal{C}_{A_{j+1}^T}$.  In the singular case
$j=\ell$, $\im(\partial_{\ell+1})=\{0\}$, a trivial vector space.
Elements of $\ker(\partial_j)\subset {\cal A}_j$ are called cycles; in
our case these are vectors in a binary linear code with the parity
check matrix $A_j$, $\ker (A_j)=\mathcal{C}_{A_j}^\perp$.  In the
singular case $j=0$, $\ker(\partial_0)={\cal A}_0$, the entire space.

Because of the orthogonality $\partial_j\partial_{j+1}=0$, all
boundaries are necessarily cycles,
$\im(\partial_{j+1})\subseteq \ker(\partial_j) \subseteq {\cal A}_j$.
The structure of the cycles in ${\cal A}_j$ that are not boundaries is
described by the $j$\,th homology group,
\begin{equation}
  {H}_j({\cal A})\equiv H(A_j,A_{j+1})=
  \ker(A_{j})/\im(A_{j+1}).
\label{eq:homo-group}
\end{equation}
Group quotient here means that two cycles [elements of $\ker (A_j)$]
that differ by a boundary [element of $\im(A_{j+1})$] are considered
equivalent; non-zero elements of $\mathcal{H}_j(\mathcal{A})$ are
equivalence classes of homologically non-trivial cycles.  Explicitly,
the equivalence of $x$ and $y$ in $\mathcal{A}_j$ implies that for
some $\alpha\in {\cal A}_{j+1}$, $y=x+ \alpha A_{j+1}^T$.  The rank of
$j$-th homology group is the dimension of the corresponding vector
space; one has
\begin{equation}
  \label{eq:homo-rank}
  k_j\equiv \rank H_j(\mathcal{A})=n_j-\rank A_j-\rank A_{j+1}  .
\end{equation}
The homological \emph{distance} $d_j$ is the minimum Hamming weight of
a non-trivial element (any representative) in the homology group
$H_j(\mathcal{A})\equiv H(A_j,A_{j+1})$,
\begin{equation}
  d_j=\min_{ 0\not\simeq x\in H_j(\mathcal{A})} \wgt x
  =\min_{x\in \ker({A}_j)\setminus \im(A_{j+1})}\wgt x.
  \label{eq:homo-distance}
\end{equation}
By this definition, $d_j\ge1$.  To address singular cases, throughout
this work we define the minimum of an empty set as infinity; $k_j=0$
is always equivalent to $d_j=\infty$.  In particular, the distance of
the homology group $H_0(\mathcal{A})$ is $d_0=1$, unless $A_1$ has
full row rank, giving $k_0=0$, in which case we get $d_0=\infty$.  In
the case of the homology group $H_{\ell}(\mathcal{A})$, the distance
$d_{\ell}$ is that of the $F$-linear  code
$\mathcal{C}^\perp_{A_\ell}$.  Again, we get $d_\ell=\infty$ if
$k_\ell=0$, which happens when $A_\ell$ has full column rank.

In addition to the homology group $H(A_j,A_{j+1})$, there is also a
\emph{co-homology} group
$\tilde{H}_j(\tilde{\cal A})=H(A_{j+1}^T,A_j^T)$ of the same rank
(\ref{eq:homo-rank}); this is associated with the \emph{co-chain
  complex} $\tilde{\cal A}$ formed from the transposed matrices
$A_j^T$ taken in the opposite order.  A quantum CSS code
with generator
matrices $G_X=A_j$ and $G_Z=A_{j+1}^T$ is isomorphic with the direct
sum of the groups $H_j$ and $\tilde{H}_j$, cf.~Eq.~(\ref{eq:css-code-decomp}),
\begin{equation}
  \label{eq:css-code-homology}
\css(A_j,A_{j+1}^T)\cong H(A_j,A_{j+1})\oplus H(A_{j+1}^T,A_j^T).
\end{equation}
The two terms correspond to $Z$ and $X$ logical operators,
respectively.  This gives for the homological distances in the chain
complex and in the co-chain complex, respectively, $d_j=d_Z$ and
$\tilde{d}_j=d_X$.

The tensor product $\mathcal{A}\times \mathcal{B}$ of two chain
complexes $\mathcal{A}$ and $\mathcal{B}$ is defined as the chain
complex formed by linear spaces decomposed as direct sums of Kronecker
products,%
\begin{equation}
  (\mathcal{A}\times \mathcal{B})_j=\bigoplus\nolimits_{i\in\mathbb{Z}}\mathcal{A}_i \otimes
  \mathcal{B}_{j-i},\label{eq:tp-spaces}
\end{equation}
with the action of the boundary operators 
\begin{equation}
  \label{eq:tp-boundary}
  \partial'''(a\otimes b)\equiv\partial' a\otimes b+(-1)^i a\otimes
  \partial'' b, 
\end{equation}
where $a\in\mathcal{A}_i$, $b\in\mathcal{B}_{j-i}$, and the boundary
operators $\partial'$, $\partial''$, and $\partial'''$ act in
complexes $\mathcal{A}$, $\mathcal{B}$, and
$\mathcal{A}\times \mathcal{B}$, respectively.  Notice that the two
terms in Eq.~(\ref{eq:tp-boundary}) are supported in different
subspaces of the expansion (\ref{eq:tp-spaces}).  When both
$\mathcal{A}$ and $\mathcal{B}$ are \emph{bounded}, that is, they
include finite numbers of non-trivial spaces, the dimension
$n_j(\mathcal{C})$ of a space $\mathcal{C}_j$ in the product
$\mathcal{C}=\mathcal{A}\times \mathcal{B}$ is%
\begin{equation}
  \label{eq:tp-nj}
n_j (\mathcal{C})=\sum\nolimits_{i} n_i (\mathcal{A}) \,n_{j-i}
  (\mathcal{B}).
\end{equation}
The homology groups of the product $\mathcal{C}={\cal A}\times {\cal B}$ are
isomorphic to a simple expansion in terms of those of ${\cal A}$ and
${\cal B}$ which is given by the K\"unneth formula,
\begin{equation}
  \label{eq:tp-Kunneth}
  H_j (\mathcal{C})\cong\bigoplus\nolimits_{i} H_i (\mathcal{A})\,
  \otimes\,H_{j-i} 
  (\mathcal{B}).
\end{equation}
One immediate consequence is that the rank $k_j(\mathcal{C})$ of the
$j$\,th homology group $H_j(\mathcal{C})$ is 
\begin{equation}
  \label{eq:tp-kj}
  k_j (\mathcal{C})=\sum\nolimits_{i} k_i (\mathcal{A}) \,k_{j-i}
  (\mathcal{B}).
\end{equation}
Such a convolution can be also written as a product of the Poincare
polynomials $p_\mathcal{A}(x)\equiv \sum_{j}k_j(\mathcal{A}) x^j$
corresponding to the two complexes,
$p_\mathcal{C}(x)=p_\mathcal{A}(x) p_\mathcal{B}(x)$.

\section{Minimal distances of certain $F$-linear CSS codes}
\label{sec:product-codes}

\subsection{Subsystem product codes and their gauge-fixed versions}
Our main tool is the map (\ref{eq:css-code-homology}) between a CSS
code and the homology groups of associated chain and co-chain
complexes.  In this section we derive the minimum distances of several
classes of CSS codes which are relevant for the analysis of the
homological distances in the tensor products of chain complexes.
Although the derivations are not technically hard, these results may
be of independent value.

The distance bounds are constructed using the following two Lemmas
which, in turn, follow from Eq.~(\ref{eq:dual-shortening}) and the
fact that for any CSS stabilizer code $ \css(H_X,H_Z)$ with logical
generator matrix $L=\diag(L_X,L_Z)$, the dual code
$\mathcal{C}_{H_X}^\perp$ coincides with the space generated by the
combined rows of $H_Z$ and $L_Z$, while $\mathcal{C}_{H_Z}^\perp$
coincides with the space generated by rows of $H_X$ and $L_X$
combined.
\begin{lemma}[$Z$-puncturing bound]
  \label{th:punctured} Consider a stabilizer code
  $\mathcal{Q}=\css(H_X,H_Z)$ with paratemeters $[[n,k,(d_X,d_Z)]]_q$
  and a qudit index set $V=\{1,2,\ldots,n\}$.  Given a partition into
  complementary sets $I\subset V$ and $J=V\setminus I$, suppose a
  logical generator matrix $L_X$ can be chosen so that none of its $k$ rows
  is supported both in $I$ and in $J$.  Let
  $Q'=\css\biglb((H_X)_I,H_Z[I]\bigrb)$ 
  and
  $Q''=\css\biglb((H_X)_J,H_Z[J]\bigrb)$ 
  be the codes whose $X$ generator matrices are shortened and $Z$
  generator matrices punctured to $I$ and $J$, respectively.  Then,
  the $Z$-distances of the three codes satisfy the inequality
  $d_Z\ge \min(d_Z',d_Z'')$.
\end{lemma}
\begin{proof}
  The case $k=0$ is trivial since it gives infinite $d_Z$; assume
  $k>0$.  The distance $d_Z$ of the code is the minimum weight in the
  set ${Q}_Z=\mathcal{C}_{H_X}^\perp\setminus \mathcal{C}_{H_Z}$ of
  all non-trivial $Z$-like codewords and their equivalent vectors.
  For any $c\in Q_Z$, the punctured vectors $c[I]$ and $c[J]$ are
  orthogonal to the rows of $(H_X)_I$ and $(H_X)_J$, respectively; the
  corresponding Pauli errors are undetectable.  Further, since
  $L_X c^T\neq0$, it is impossible that $c[I]$ be orthogonal to the
  rows of $(L_X)_I=L_X[I]$ and at the same time $c[J]$ be orthogonal
  to the rows of $(L_X)_J=L_X[J]$.  Therefore, at most one of the
  vectors $c[I]$ and $c[J]$ can be trivial in the corresponding code.

  Now, consider the identity $\wgt c[I] + \wgt c[J] = \wgt c > 0$. The
  punctured pieces $c[I]$ and $c[J]$ contribute to the distances
  $d_Z'$ and $d_Z''$ respectively only if the corresponding vectors
  are non-trivial. Let $d(c)$ equal infinity if $ c$ is trivial in
  $\mathcal{Q}$, and $\wgt c \ge 1$ otherwise, and define similar
  functions $d'(c)$ and $d''(c)$ for vectors corresponding to
  undetectable errors in $Q'$ and $Q''$, respectively.  Then,
  $d_Z'\le \min_{c\in\mathcal{Q}_Z} \mathop{d'}(c[I])$ and
  $d_Z''\le \min_{c\in\mathcal{Q}_Z} \mathop{d''}(c[J])$.  The stated
  result is obtained by minimizing the inequality
  $\min\biglb( \mathop{d'}(c[I]), \mathop{d''}(c[J])\bigrb)\le d(c)$
  over all $c \in Q_Z$.
\end{proof}

\begin{lemma}[$Z$-shortening bound]
  \label{th:shortened}
  Consider a stabilizer code $\mathcal{Q}=\css(H_X,H_Z)$ with the set
  $V$ indexing its variable nodes.  For any index set $I\subset V$,
  let $Q'=\css\biglb(H_X[I],(H_Z)_I\bigrb)$ be the code whose $X$
  generator matrix is punctured and $Z$ generator matrices shortened
  to $I$.  Then (i) the $Z$-distances of the original code does not
  exceed that of $Q'$, $d_Z\le d_Z'$.  (ii) This inequality is
  saturated if the support of a minimum-weight codeword in $Q_Z$ is
  contained in $I$.
\end{lemma}
\begin{proof}
  This follows from the facts that (a) any codeword in $Q_Z'$ is also in
  $Q_Z$, and (b) that any codeword in $Q_Z$ which is supported on $I$
  is also in $Q_Z'$.
\end{proof}

We now consider several ``product'' codes related to the
subsystem code ${\cal Q}^{\rm subs}=\css(G_X,G_Z)$ with the gauge generator matrices
\begin{equation}
  \label{eq:gen-proj-subs}
  \strut\!\!\! G_X=\left(
    \begin{array}[c]{c}
      H_X^A\otimes I(n_B)\\
      I(n_A)\otimes H_X^B
    \end{array}
  \right),\;
  G_Z=\left(
    \begin{array}[c]{c}
      H_Z^A\otimes I(n_B)\\
      I(n_A)\otimes H_Z^B
    \end{array}
  \right),\!
\end{equation}
constructed in terms of generator matrices of a pair of stabilizer
codes $Q_A=\css(H_X^A,H_Z^A)$ and $Q_B=\css(H_X^B,H_Z^B)$ with
parameters $[[n_A,k_A,(d_X^A,d_Z^A)]]_q$ and
$[[n_B,k_B,(d_X^B,d_Z^B)]]_q$, respectively.

\begin{lemma}[Subsystem product code]
  \label{th:subs-code-H}
  Denote $L_X^A$, $L_Z^A$ and $L_X^B$, $L_Z^B$ the logical generator
  matrices of the CSS stabilizer codes $\css(H_X^A,H_Z^A)$ and
  $\css(H_X^B,H_Z^B)$, respectively, chosen so that
  \begin{equation}
    L_X^A (L_Z^A)^T=I(k_A),\quad  L_X^B (L_Z^B)^T=I(k_B).
    \label{eq:A-B-logical} 
\end{equation}
Then the \emph{subsystem product code} with CSS gauge generator
matrices (\ref{eq:gen-proj-subs}) has logical generator matrices
  \begin{equation}
    \label{eq:log-proj}
    L_X=L_X^A\otimes L_Z^B,\quad L_Z=L_Z^A\otimes L_Z^B,  
  \end{equation}
  and stabilizer generator matrices
  \begin{equation}
    \label{eq:stab-proj}
    H_X=\left(
      \begin{array}[c]{c}
        H_X^A\otimes H_X^B\\
        H_X^A\otimes L_X^B\\
        L_X^A\otimes H_X^B
      \end{array}
    \right),\quad
    H_Z=\left(
      \begin{array}[c]{c}
        H_Z^A\otimes H_Z^B\\
        H_Z^A\otimes L_Z^B\\
        L_Z^A\otimes H_Z^B
      \end{array}
    \right).
  \end{equation}
\end{lemma}
\begin{proof}
  Matrices
  \begin{equation}
    \label{eq:2}
    P_A=\left(\begin{array}[c]{c}L_Z^A\\ H_Z^A\end{array}\right),\quad 
    P_B=\left(\begin{array}[c]{c}L_Z^B\\ H_Z^B\end{array}\right)
  \end{equation}
  are the parity check matrices for the classical $F$-linear codes
  with generator matrices $H_X^A$ and $H_X^B$, respectively.  Thus, a
  classical code with generator matrix $G_X$ in
  Eq.~(\ref{eq:gen-proj-subs}) has a parity check matrix
  $P_A\otimes P_B$.  Out of the four row blocks of the latter
  matrix, only rows of $L_Z=L_Z^A\otimes L_Z^B$ are linearly
  independent from the rows of $G_Z$, as can be verified by taking
  scalar products with the rows of $L_X$.  The remaining row blocks
  can be readily seen as linear combinations of the rows of $G_Z$;
  they form the matrix $H_Z$.  The proof for $L_X$ and $H_X$ is
  similar.
\end{proof}

\begin{theorem}[Concatenated-stabilizer CSS code]
  \label{th:concat-stab-code}
  Let $\mathcal{Q}_A$ and $\mathcal{Q}_B$ be two $F$-linear CSS
  stabilizer codes used to define matrices (\ref{eq:gen-proj-subs}),
  with logical generator matrices (\ref{eq:A-B-logical}).  Use $n_B$
  copies of the code $\mathcal{Q}_A$, with logical operators used as
  qudits for the outer code, to form a \emph{concatenated-stabilizer
    code} $\overline {\cal Q}$ with CSS generator matrices
  \begin{equation}
    \label{eq:concatenated-CSS}
    \overline{H}_X=\left(
      \begin{array}[c]{c} H_X^A\otimes I(n_B)\\ L_X^A\otimes H_X^B        
      \end{array}\right),\;
    \overline{H}_Z=\left(
      \begin{array}[c]{c} H_Z^A\otimes I(n_B)\\ L_Z^A\otimes H_Z^B        
      \end{array}\right).
  \end{equation}
  The logical generator matrices of thus constructed code are given by
  Eq.~(\ref{eq:log-proj}), and the parameters are given by the
  corresponding products
  $[[n_An_B, k_Ak_B, (d_X^Ad_X^B,d_Z^A d_Z^B)]]_q$.
\end{theorem}

\begin{proof}
  It is easy to check that $\overline{H}_X\overline{H}_Z^T=0$; this is
  a stabilizer code.  Similarly, we get 
  $\overline{H}_XL_Z^T=0$, $\overline{H}_ZL_X^T=0$, 
  $L_XL_Z^T=I(k_A)\otimes I(k_B)$, and the matrix ranks
  \begin{eqnarray}
    \label{eq:rank-concatenated-stab}
    \rank \overline{H}_X&=&\rank H_X^A\, n_B+k_A\rank H_X^B,\\
    \rank \overline{H}_Z&=&\rank H_Z^A\, n_B+k_A\rank H_Z^B,\\
    \rank L_X&=&\rank L_Z=k_Ak_B;
  \end{eqnarray}
  these expressions add up to the code length $n_An_B$.  This verifies
  the CSS construction and the number of encoded qudits $k=k_Ak_B$.
  The case $k=0$ is trivial; in the following, assume $k>0$.  To
  construct the upper distance bounds, e.g., $d_Z\le d_Z^Ad_Z^B$,
  consider pairs of conjugated codewords $a,a'$ and $b,b'$ in
  $\mathcal{Q}_A$ and $\mathcal{Q}_B$, respectively, where $a$ and $b$
  are $Z$-like with $\wgt a=d_Z^A$, $\wgt b=d_Z^B$, and
  $a' a^T=b' b^T=1$.  Then the vector $c=a\otimes b$ of weight
  $d_Z^Ad_Z^B$ satisfies $\overline H_X c^T=0$.  Further, its dual
  $a'\otimes b'$ is orthogonal to the rows of $\overline{H}_Z$, which
  implies that $c$ cannot be a linear combination of the rows of
  $\overline{H}_Z$.  Taken together, this proves
  $c\in \overline{\cal Q}_Z$, thus its weight gives a valid upper
  bound on $d_Z$.

  To construct a matching lower distance bound, assume there is a
  non-trivial codeword $c\in\overline{\cal Q}_Z$ such that $\wgt(c)<d_Z^Ad_Z^B$.
  This implies $\overline{H}_Xc^T=0$, and also that $c$ must be linearly
  independent from the rows of $\overline{H}_Z$.  Let $e_j\in F^{n_B}$,
  $j\in\{1,\ldots, n_B\}$ be vectors with all zero components except a
  one at position $j$.  Consider a decomposition
  \begin{equation}
    \label{eq:tensor-product-decomposition}  
    c=\sum_j a_j\otimes e_j,\quad \text{where}\quad a_j\in F^{n_A}.
  \end{equation}
  From the upper row blocks of the generators
  (\ref{eq:concatenated-CSS}), each non-zero $a_j$ must either be a
  non-trivial $Z$-like vector in the code $\mathcal{Q}_A$, or a linear
  combination of the rows of $H_Z^A$.  This implies that any non-zero
  $a_j$ such that $\wgt (a_j)<d_Z^A$ can be removed from $c$ (set to
  zero) without any other changes; the resulting vector $c'$ should
  remain in the code as the two vectors are degenerate with
  respect to $\mathcal{C}_{\overline{H}_Z}$.  This vector has weight
  $\wgt(c')\le \wgt(c)<d_Z^Ad_Z^B$, and any non-zero component $a_j$
  in its expansion (\ref{eq:tensor-product-decomposition}) has weight
  $d_Z^A$ or larger.  Let $J\subset \{1,2,\ldots,n_B\}$ be the set of
  positions $j$ corresponding to non-zero $a_j$ in the expansion of
  $c'$.  By this logic,
  \begin{equation}
    d_Z^Ad_Z^B> \wgt(c')=\sum_{j\in J}\wgt(a_j)\ge d_Z^A|J|;
    \label{eq:projection-decomposition}  
\end{equation}
  the total number of positions in $J$ satisfies $|J|<d_Z^B$.  Denote
  $V_A=\{1,2,\ldots, n_A\}$ and $I\equiv V_A\otimes J$; the punctured
  vector $c'[I]$ preserves all non-zero positions in $c'$.  Thus,
  $c'[I]$ should be in the code
  $\mathcal{Q}'=\css\biglb(\overline{H}_X[I],(\overline{H}_Z)_I\bigrb)$,
  see Lemma \ref{th:shortened}.  By construction, the matrices
  $\overline{H}_X [I]$ and $(\overline {H}_Z)_I$ have the same
  structure (\ref{eq:concatenated-CSS}), except the code
  $\mathcal{Q}_B$ is replaced with
  $\mathcal{Q}_B'=\css\biglb(H_X [J],(H_Z)_J\bigrb)$ of length $|J|$.
  This latter code also satisfies Lemma \ref{th:shortened}; we expect
  the corresponding distance to serve as an upper bound to $d_Z^B$.
  However, since its length $|J|<d_Z^B$, the only possibility is for
  the code $\mathcal{Q}_B'$ to encode no qubits, $k_B'=0$.
  Necessarily, the code ${Q}'$ also has $k'=k_Ak_B'=0$, which makes
  the initial assumption about the existence of the codeword $c$
  invalid; this proves $d_Z=d_Z^Ad_Z^B$.
\end{proof}
\subsection{Bounds on the minimal distance}
Notice that rows of $H_Z$ in Eq.~(\ref{eq:stab-proj}) are linear
combinations of rows of $\overline{H}_Z$ in
Eq.~(\ref{eq:concatenated-CSS}), whose rows are, in turn, linear
combinations of rows of $G_Z$ in Eq.~(\ref{eq:gen-proj-subs}).
Similar relation exists between the corresponding $X$ matrices.  As a
result, there is a sequence of inclusions,
\begin{equation}
  \label{eq:inclusion-seq}
  {\cal C}_{G_X}^\perp\setminus   {\cal C}_{H_Z}\subseteq   {\cal
    C}_{\overline{H}_X}^\perp\setminus   {\cal
    C}_{\overline{H}_Z}\subseteq   {\cal C}_{H_X}^\perp\setminus   {\cal
    C}_{G_Z}, 
\end{equation}
which implies a sequence of inequalities for the three related codes:
\begin{equation}
  d_Z(G_X,H_Z)\ge d_Z(\overline{H}_X,\overline{H}_Z)\ge
  d_Z(H_X,G_Z),
\end{equation}
where, e.g., $  d_Z(G_X,H_Z)$ is the $Z$-distance in the code
$\css(G_X,H_Z)$.  

On the other hand, from  linear relations between the rows of
matrices involved, Lemma \ref{th:subs-code-H} and Theorem
\ref{th:concat-stab-code}, it follows that all of the three codes in
Eq.~(\ref{eq:inclusion-seq}) are gauge-fixed versions of the subsystem
code with the generators (\ref{eq:gen-proj-subs}).  They share the
logical generator matrices (\ref{eq:log-proj}), which implies a common
upper bound $d_Z\le d_Z^Ad_Z^B$; the proof is similar to that in
Statement \ref{th:concat-stab-code}.  We get
\begin{eqnarray}
  \label{eq:symmetric-distance}
  d_Z(G_X,G_Z)&=&  d_Z(H_X,G_Z)\le d_Z^Ad_Z^B,
  \\
  d_Z(G_X,H_Z)&=& d_Z(\overline{H}_X,\overline{H}_Z)=d_Z^Ad_Z^B.
\end{eqnarray}

Unfortunately, we are not able to get the exact values for the
$Z$-distances in the l.h.s.\ of Eq.~(\ref{eq:symmetric-distance}).  It
is clear that the general upper bound (\ref{eq:symmetric-distance}) is
sharp.  In particular, the upper bound is saturated whenever one of
the codes has distance one.  This follows from the following two lower
bounds which we adapted from Ref.~\onlinecite{Audoux-Couvreur-2017}.

\begin{statement}[Lower distance bound I]
  \label{th:loose-lower-bound}
  Consider an $F$-linear code $\css(H_X,G_Z)$ with stabilizer
  generator matrices $H_X$ and $G_Z$ given by
  Eqs.~(\ref{eq:stab-proj}) and (\ref{eq:gen-proj-subs}),
  respectively.  (a) The corresponding $Z$-distance satisfies the
  inequality%
  \begin{equation}
    d_Z(H_X,G_Z)\ge \max(d_Z^A,d_Z^B).\label{eq:loose-lower-bound}
  \end{equation}
  (b) In addition, assume that $d_Z^A>1$.  Then, with
  $F=\mathbb{F}_q$,%
  \begin{equation}
    d_Z(H_X,G_Z)\ge{q\over q-1}\, d_Z^B.\label{eq:loose-lower-bound-two}
  \end{equation}
\end{statement}
The proof is based on the following Lemma from
Ref.~\onlinecite{Audoux-Couvreur-2017}:
\begin{lemma}[Lower distance bound II]
  \label{th:lower-bnd-lemma}
  Consider an $F$-linear stabilizer code ${\cal Q}=\css(H_X,G_Z)$ with
  generator matrices $H_X$ and $G_Z$ in Eqs.~(\ref{eq:stab-proj}) and
  (\ref{eq:gen-proj-subs}), respectively.  Given $a\in {\cal Q}_A^X$,
  consider a set $\Omega_A(a)=\{x_1,x_2,\ldots,x_N\}$ of vectors
  degenerate with $a$ with respect to ${\cal C}_{H_X^A}$, such that
  each $i\in \{1,2,\ldots,n_A\}$ is in the support of no more than $K$
  of these vectors.  Then, for any $Z$-like codeword
  $c\in {\cal Q}^Z$ such that $[a\otimes I({n_B})]\, c^T\neq 0$,
  \begin{equation}
    \label{eq:lower-bnd-lemma}
    \wgt(c)\ge \left\lceil {N\over K}\, d_Z^B\right\rceil.
  \end{equation}
\end{lemma}
\begin{proof}
  Given $c$ in Eq.~(\ref{eq:lower-bnd-lemma}), consider an expansion
  $$
  c=\sum_{j=1}^{n_A} f_j\otimes b_j, \quad b_j\in F^{n_B},
  $$
  where components of $f_j\in F^{n_A}$ are all zero except for
  $f_j[j]=1$, $j\in\{1,\ldots,n_A\}$.  By assumption, the dot-product
  $a_i\otimes I(n_B)$ with $c$ is non-zero; for any
  $a_i\in\Omega_A(a)$,%
  $$
  x_i^T\equiv (a_i\otimes I_{n_B} ) \,c^T=\sum_j a_i[j]\, b_j^T.
  $$
  It is easy to check that the resulting vector $x_i\in F^{n_B}$
  satisfies $H_X^B x_i^T=0$, while $L_X^Bx_i^T\neq0$.  That is, $x_i$
  is in ${\cal Q}_B^Z$, so that $\wgt(x_i)\ge d_Z^B$.  Let us now sum
  the weights of vectors $x_i$ corresponding to all elements of
  $\Omega_A(a)$,
  \begin{eqnarray*}
    N d_Z^B  \le  \sum_{i=1}^N \wgt(x_i)
    &\le& \sum_{j=1}^{n_A} \sum_{i=1}^N \wgt(a_i[j]\,b_j^T)
    \\  
    &\le & K\sum_{j=1}^{n_A}  \wgt(b_j) =  K \wgt(c),
  \end{eqnarray*}
which gives Eq.~(\ref{eq:lower-bnd-lemma}) since $\wgt(c)$ is an integer.
\end{proof}

\begin{proof}[Proof of Statement \ref{th:loose-lower-bound}]
  Both (a) and (b) are trivial if $k_Ak_B=0$; assume otherwise below.
  (a) The construction is symmetric with respect to constituent codes
  ${\cal Q}_A$ and ${\cal Q}_B$; without limiting generality assume
  $d_Z^B\ge d_Z^A$.  Use the set $\Omega_A(a)=\{a\}$ in Lemma
  \ref{th:lower-bnd-lemma} with $N=K=1$, which proves
  $d_Z(H_X,G_Z)\ge d_Z^B$.  (b) The condition $d_Z^A>1$ implies
  that any all-zero column in $H_X^A$ (say, at position $i\le n_A$)
  must be matched by a row (or a linear combination of rows) of
  $H_Z^A$ with the only non-zero element at $i$.  This guarantees that
  any $X$-like codeword $a$ has no support at such position(s).  For
  any $a\in {\cal Q}_X^A$, consider the set $\Omega\equiv\Omega_A(a)$
  of size $N=q^{\rank H_A^X}$ which contains all vectors degenerate
  with $a$.  For any $i\le n_A$, the set of characters
  $\Omega[i]\equiv \{x[i]:\forall x\in\Omega_A(a)\}$ either contains
  all zeros, or contains equal numbers of all elements of $F$---this
  can be seen by considering a generating matrix with all except one
  row not supported on $i$.  For such a set, $K=(q-1)q^{(\rank H_A^X-1)}$,
  which proves Eq.~(\ref{eq:loose-lower-bound-two}). 
\end{proof}

Another application of Lemma \ref{th:lower-bnd-lemma} is demonstrated by
the following
\begin{example}
  \label{ex:single-qubit-cyclic}
  Let $\mathcal{Q}^A=\css(H_X^A,H_Z^A)$ be a single-qubit encoding
  (consta)cyclic CSS code with parameters $[[n_A,1,(d_X^A,d_Z^A)]]$.
  Then, for any $\mathcal{Q}^B=\css(H_X^A,H_Z^A)$, the $Z$-distance of
  the product code $\css(H_X,G_Z)$ with stabilizer generator matrices
(\ref{eq:stab-proj}) and (\ref{eq:gen-proj-subs}) satisfies
  \begin{equation}
    \label{eq:single-qubit-cyclic}
    d_Z(H_X,G_Z)\ge \lceil n_A d_Z^B/d_X^A\rceil.
  \end{equation}
\end{example}
\begin{proof}
  Use Lemma \ref{th:lower-bnd-lemma} with $\Omega_A(a)$ a set of size
  $N=n_A$ constructed by shifting an $X$-like minimum-weight codeword
  $a\in \mathcal{Q}_A^X$, $\wgt(a)=d_X^A$ by $0$, $1$, \ldots, $n_A-1$
  positions.  The resulting vectors $x_i\in\Omega_A(a)$ cannot be
  linear combinations of rows $H_X^A$, or else the original vector $a$
  would be too, thus they must be in the code.  Since $k_A=1$, they
  must be degenerate with $a$.  The lower bound
  (\ref{eq:single-qubit-cyclic}) is obtained if we notice that for
  this set, $K=d_X^A$.
\end{proof}

The discussed lower distance bounds are pretty far from the generic
upper bound (\ref{eq:symmetric-distance}).  On the other hand, at
least in the binary case, it is not easy to construct an example of a
subsystem product code with the distance strictly below the upper
bound.  Discovering such examples is dramatically simplified with the
help of the ansatz in the following Theorem
\ref{ex:upper-distance-bnd-nA}, a generalization of the construction
based on the homological product of Steane's $[[7,1,3]]$ code with
itself\cite{Bravyi-Hastings-2013,Audoux-Couvreur-2017} (see Example
\ref{ex:square-Steane} below)
\begin{theorem}[$X$--$Z$-symmetric product codes]
  \label{ex:upper-distance-bnd-nA}
  Consider codes ${\cal Q}_A=\css(H_X^A,H_Z^A)$ and
  ${\cal Q}_B=\css(H_Z^A,H_X^A)$ with $X$ and $Z$ generator matrices
  interchanged.  The distances of the corresponding subsystem product
  code $\css(G_X,G_Z)$ with generators (\ref{eq:gen-proj-subs}) satisfy  
  \begin{equation}
    d_X(G_X,G_Z)\le n_A, \quad d_Z(G_X,G_Z)\le n_A.
    \label{eq:upper-distance-bnd-nA}  
 \end{equation}
  The inequality (\ref{eq:symmetric-distance}) becomes strict
  if $n_A<d_Z^Ad_Z^B\equiv d_Z^Ad_X^A$.
\end{theorem}
\begin{proof}
  The construction is symmetric with respect to $X$ and $Z$ parts of
  ${\cal Q}_A$; it is sufficient to prove the bound for
  $d_Z\equiv d_Z(G_X,G_Z)=d_Z(H_X,G_Z)$ with $H_X$ in
  Eq.~(\ref{eq:stab-proj}).  We have $n_A=n_B$; consider a vector
  $c=\sum_{j=1}^{n_A}e_j\otimes e_j$ of weight $n_A$, where $e_j$ are
  weight-one vectors as in
  Eq.~(\ref{eq:tensor-product-decomposition}).  Using
  Eq.~(\ref{eq:A-B-logical}) and the orthogonality between the rows of
  remaining $X$ and $Z$ generator matrices, verify that $H_Xc^T=0$
  while $L_Xc^T\neq0$.  Thus, $c$ is a valid $Z$-like codeword in
  $\css(H_X,G_Z)$ and $d_Z\le n_A$.
\end{proof}
It is known that long CSS codes with distances scaling linearly
with the code length $n$ exist\cite{Calderbank-Shor-1996}.  For a pair
of such codes, the generic upper bound (\ref{eq:symmetric-distance})
has asymptotic scaling $d\le {\cal O}(n_An_B)$, linear in the length
of the product code.  On the other hand, the upper bound for the
corresponding $X$--$Z$-symmetric product codes, see Theorem
\ref{ex:upper-distance-bnd-nA}, gives $d\le n_A$, a square root of the
length of the product code.  Thus, we can not expect the generic upper
bound to be saturated.  The following explicit Examples demonstrate
that such a saturation does not happen for any finite field $F$.

\begin{example}
  \label{ex:odd-base-code}
  For any field $F=\mathbb{F}_q$ with $q\equiv 2t+1$ odd, consider a
  $[[3,1,(2,2)]]_q$ code with CSS generators $H_X=(1,1,1)$,
  $H_Z=(t,t,1)$.  The corresponding $X$--$Z$-symmetric product code in
  Theorem \ref{ex:upper-distance-bnd-nA} has distances $d_X=d_Z=3$,
  smaller than the upper bound (\ref{eq:symmetric-distance}).  For
  $q=3$, this saturates the lower bound
  (\ref{eq:loose-lower-bound-two}).
\end{example}

\begin{example}
  \label{ex:dbl-even}
  For any $q=2^m$ with $m$ even, so that $r\equiv (q-1)/3$ be an
  integer, consider a stabilizer code $[[3,1,(2,2)]]_q$ with cyclic
  $H_X^A$ and constacyclic $H_Z^B$ generators,
\begin{equation}
  \label{eq:counterexample-dbl-even}
  H_X^A=\left(
    \begin{array}[c]{ccc}
      1&1&1
    \end{array}\right), \quad
    H_Z^A=\left(
    \begin{array}[c]{ccc}
      1&x^r&x^{2r}
    \end{array}\right),
\end{equation}
where $x\in\mathbb{F}_q$ is a primitive element, i.e., $x^{q-1}=1$.
Construct an $X$--$Z$-symmetric product code as in Theorem
\ref{ex:upper-distance-bnd-nA}.  Combining
Eq.~(\ref{eq:upper-distance-bnd-nA})  with the lower bound
(\ref{eq:single-qubit-cyclic}) again gives $d_Z(G_X,G_Z)=3$, smaller
than $d_Z^Ad_Z^B=d_Z^Ad_X^A=4$.
\end{example}
\begin{example}[Square of Steane's code\cite{Bravyi-Hastings-2013,Audoux-Couvreur-2017}]
  \label{ex:square-Steane}
  For any $q=2^m$, $m\in\mathbb{N}$, consider a pair of identical
  cyclic codes $[[7,1,(3,3)]]_q$ with stabilizer generator polynomials
  $h_X^A(x)=h_Z^B(x)=1+x^2+x^3+x^4$.  Combination of the $X$--$Z$
  symmetric product construction from Theorem
  \ref{ex:upper-distance-bnd-nA} and the lower bound
  (\ref{eq:single-qubit-cyclic}) gives $d_Z(G_X,G_Z)=7$, smaller than
  $d_Z^Ad_Z^B=9$.
\end{example}

\subsection{Previously known constructions}
In the remainder of this Section, we discuss several existing code
families which can be described as special cases of the subsystem
product code construction in Lemma \ref{th:subs-code-H}, or as
gauge-fixed versions of such codes.

The first such family, homological product codes from
Refs.~\onlinecite{Bravyi-Hastings-2013,Audoux-Couvreur-2017}, is based
on square nilpotent matrices such that $\delta^2=0$, with elements
from a field $F=\mathbb{F}_q$ with $q=2^m$, $m\in\mathbb{N}$.  Such a
matrix $\delta$ and its transposed $\delta^T$ can be used to construct
the stabilizer code $\css(\delta,\delta^T)$ and its symmetric
$\css(\delta^T,\delta)$.  Alternatively, stabilizer generators of a
CSS code with $\rank{H_X}=\rank{H_Z}$ can be used to form such a
nilpotent matrix, $\delta=H_X^TMH_Z$, where $M$ is a matrix of
appropriate dimensions chosen to preserve the rank of the product.
\begin{example}[Homological product codes]
  \label{ex:homo-prod-code}
  For $q=2^m$, $m\in\mathbb{N}$, consider a pair of $F$-linear
  stabilizer codes $\mathcal{Q}^\mu=\css(\delta_\mu,\delta_\mu^T)$
  with parameters $[[n_\mu,k_\mu,(d_X^\mu,d_Z^\mu)]]_q$ based on
  nilpotent matrices $\delta_\mu$, where $\mu\in\{A,B\}$.  Then the
  matrix $\delta_C=I(n_A)\otimes \delta_B+\delta_A\otimes I(n_B)$ is
  also nilpotent.  The corresponding code $\css(\delta_C,\delta_C^T)$
  has logical generator matrices given by
  Eq.~(\ref{eq:gen-proj-subs}), and parameters
  $[[n_An_B,k_Ak_B,(d_X^C,d_Z^C)]]_q$, where, e.g.,
  \begin{equation}
    \label{eq:hom-prod-distance-bnd}
    d_Z(G_X,G_Z)\le d_Z^C\le d_Z^Ad_Z^B.
  \end{equation}
\end{example}
\begin{proof}
  It is easy to check that the logical generator matrices are given by
  Eq.~(\ref{eq:log-proj}); the upper bound on the distance follows.
  On the other hand, rows of $\delta_C$ and $\delta_C^T$,
  respectively, are linear combinations of the rows of $G_X$ and
  $G_Z$, see Eq.~(\ref{eq:gen-proj-subs}).  This implies that the
  stabilizer code defined by this matrix is a gauge-fixed version of
  the subsystem product code $\css(G_X,G_Z)$, which gives the lower
  bound. 
\end{proof}

As before, the upper distance bound is sharp, but it is not
necessarily saturated.  In particular, an
example\cite{Bravyi-Hastings-2013} can be constructed along the lines
of Example \ref{ex:square-Steane}, as a homological product code
combining two Steane's codes with identical symmetric nilpotent
matrices $\delta$.  Such a code has distance $d=7$, while the  the
upper bound in Eq.~(\ref{eq:hom-prod-distance-bnd}) gives $d\le 9$.

Our last example shows that subsystem product codes and the
corresponding gauge-fixed codes from Lemma \ref{th:subs-code-H} can be
seen as a generalization of subsystem hypergraph-product codes and
corresponding gauge-fixed codes recently constructed by Li and
Yoder\cite{Li-Yoder-2020} which are, in turn, a generalization of
Bacon-Shor\cite{Bacon-subs-2006} and Shor's\cite{Shor-FT-1996} codes,
respectively.  The Li--Yoder construction is based
on a pair of classical codes, it is similar but not identical to those
in Refs.~\onlinecite{Li-etal-Brown-2018,Yoder-2019}.  Namely, the
gauge and stabilizer generator matrices can be obtained from
Eqs.~(\ref{eq:gen-proj-subs}) and (\ref{eq:stab-proj}) by considering
the classical codes as degenerate quantum codes with empty $H_Z^A$ and
$H_X^B$ matrices.
\begin{example}[Subsystem QHP codes\cite{Li-Yoder-2020}]
  \label{ex:bacon-shor}
  Given a pair of $F$-linear classical codes with parameters
  $[n_\mu,k_\mu,d_\mu]_q$, parity check matrices $P_\mu$, and
  generator matrices $Q_\mu$, where $\mu\in\{A,B\}$, consider a
  subsystem code $\css(G_X,G_Z)$ with gauge generator matrices
  \begin{equation}
    G_X=(P_A\otimes I_{n_B}),\quad G_Z=(I_{n_A}\otimes P_B).
    \label{eq:subs-simplified-G}  
\end{equation}
  The corresponding stabilizer generator matrices are
  \begin{equation}
  H_X=\left(P_A\otimes Q_B\right),\quad
  H_Z=\left(Q_A\otimes P_B\right).\label{eq:subs-simplified-H}
\end{equation}
  Assuming $k_Ak_B>0$, the parameters of the subsystem and both
  gauge-fixed codes are $[[n_An_B,k_Ak_B,(d_A,d_B)]]_q$.
\end{example}
The parameters of the codes follow from Theorem
(\ref{th:concat-stab-code}) where we should use $d_X^A=d_Z^B=1$.  In
particular, we get the original Bacon-Shor (BS) and Shor's codes if we
take repetition codes for both classical codes.   

We also notice that a subsystem product code constructed from a BS
code and a repetition code coincides with the 3-dimensional BS code as
proposed by Napp and Preskill\cite{Napp-Preskill-2013} (this
construction differs from the 3D code originally suggested by
Bacon\cite{Bacon-subs-2006}).  Napp \& Preskill construction can be
 seen as a three-fold subsytem product of repetition codes, and
can be generalized to higher dimensions.  However, it is easy to check
that these single-qubit encoding codes are just rearrangements of
conventional BS codes from a 2D lattice to higher dimensions.  The only
differences are the measurement redundancy and local connectivity of
neighboring qubits, as defined by the specific sets of gauge generators
used in the construction.

\section{Homological distances in tensor products of chain
  complexes}
\label{sec:complexes}
Example \ref{ex:bacon-shor} may serve as a nice introduction to the
subject of this section.  Indeed, Bacon-Shor code can be obtained from
Kitaev's toric code by erasing qubits on all vertical (or all
horizontal) bonds.  The latter code corresponds exactly to a
CW-complex associated with a square lattice with periodic boundary
conditions---a tensor product of two cycle graphs.  More general gauge
generator matrices (\ref{eq:subs-simplified-G}) can be seen as a
result of erasing one of the blocks in a QHP
code\cite{Tillich-Zemor-2009,Tillich-Zemor-2014} with stabilizer
generator matrices
\begin{eqnarray}
  \nonumber 
  H_X&=&(P_A\otimes I_{n_B}|I_B\otimes P_B^T),\\ 
  H_Z&=&(I_{n_A}\otimes P_B|-P_A^T\otimes I_A),
  \label{eq:qhp-generators}
\end{eqnarray}
where the dimensions of the identity matrices $I_A$ and $I_B$ match
the numbers of rows in the two check matrices.  The matrices $H_X$ and
$H_Z^T$ correspond exactly to the boundary operator matrices in a
product of the chain complexes $\mathcal{K}(P_A)$ and
$\mathcal{K}(P_B^T)$.  In this section we consider tensor products of
general bounded $F$-linear chain complexes.  The corresponding
boundary operators, see Eq.~(\ref{eq:Cj}) below, have row- and
column-blocks with the structure of the gauge generator matrices
(\ref{eq:gen-proj-subs}).  

\subsection{Main results for $F$-linear chain complexes}
\label{sec:results}

Our main result is the expression for the homological distance in a
tensor product of two bounded chain complexes of finite-dimensional
vector spaces over a finite field $F$, where one of the complexes
contains just two non-trivial spaces.  Specifically, let $\mathcal{A}$
be such a complex of any length specified in terms of boundary
operators $\partial_j:\mathcal{A}_{j-1}\leftarrow \mathcal{A}_j$
defined explicitly as matrices, $\partial_j=A_{j}$ such that
$A_j A_{j+1}=0$, and $\mathcal{B}$ a complex with just two non-trivial
spaces $\mathcal{B}_0$ and $\mathcal{B}_1$ and a single non-trivial
boundary operator (matrix)
$B_1: \mathcal{B}_0\leftarrow \mathcal{B}_1$ mapping between them.
Then, the homological distance $d_j(\mathcal{C})$ for the $j$\,th
homology group in the tensor product
$\mathcal{C}=\mathcal{A}\times \mathcal{B}$ of the two complexes is
\begin{equation}
  \label{eq:result-thm}
  d_j(\mathcal{C})
  =\min\biglb( d_j(\mathcal{A}) d_0(\mathcal{B}),
  d_{j-1}(\mathcal{A}) d_1(\mathcal{B})\bigrb). 
\end{equation}
This is a generalization of the identical expression for the tensor
product of binary chain complexes from
Ref.~\onlinecite{Zeng-Pryadko-2018}.

There is actually a stronger statement which concerns the homological
distance $d_{j}(\mathcal{C}_{i,j-i})$ after a projection onto a single
subspace $\mathcal{C}_{i,j-i}=\mathcal{A}_i\otimes \mathcal{B}_{j-i}$,
where $j-i\in \{0,1\}$.  Here, a chain complex with the space
$\mathcal{C}_j$ reduced to its subspace has modified boundary
operators $\partial_i'$ and $\partial_{i+1}'$.  The latter is defined
as a composition of a projector $P$ and the original boundary operator
$\partial_{i+1}$, $\partial_{i+1}'\equiv P\partial_{i+1}$, where
$P^2=P$ and the image of $P$ is the subspace of interest.  The
modified boundary operator $\partial_i'$ is defined to ensure the
composition to vanish, $\partial_i'\,\partial_{i+1}'=0$.  For thus defined 
chain complex ${\cal C}'_{i,j-i}$ with the space $\mathcal{C}_j$ in
the original product complex ${\cal C}$ projected to its subspace
$\mathcal{C}_{i,j-i}$, the homological distance at level $j$ is given
by one term only,%
\begin{equation}
  \label{eq:result-proj}
  d_{j}(\mathcal{C}'_{i,j-i})=d_i(\mathcal{A}) d_{j-i}(\mathcal{B}),\quad
  j-i\in \{0,1\}.  
\end{equation}

Our third result concerns with the minimal distance in a tensor
product of two arbitrary-length chain complexes of vector spaces over
a finite field $F$.  Here the upper bound on the homological distance
reads
\begin{equation}
  \label{eq:upper-general}
  d_{j}(\mathcal{C})\le \min_{i\in\mathbb{Z}} d_{i}(\mathcal{A})
  d_{j-i}(\mathcal{B}). 
\end{equation}
A lower bound for the same distance $d_j(\mathcal{C})$ can be
constructed by projecting onto the individual product spaces
$\mathcal{A}_i\otimes \mathcal{B}_{j-i}$, $i\in\mathbb{Z}$, whose
direct sum gives the degree-$j$ space ${\cal C}_j$ in the product
complex.  This gives
$ d_j (\mathcal{C})\ge \min_i d(\mathcal{C}_{i,j-i}')$.  The result of
the projection can be seen as an $F$-linear quantum subsystem code
with CSS gauge generator matrices in the product form
(\ref{eq:gen-proj-subs}),%
\begin{equation}
  \label{eq:subsystem-product}
 G_X=\left(
    \begin{array}[c]{c}
      I({a_i})\otimes {B}_{j-i}\\ {A}_i\otimes I(b_{j-i})
    \end{array}\right),\;
  G_Z=\left(
    \begin{array}[c]{c}
      I({a_i})\otimes {B}_{j-i+1}^T\\ {A}_{i+1}^T\otimes I(b_{j-i})
    \end{array}\right),
\end{equation}
where $I(a)\equiv I_a$ is the size-$a$ identity matrix, and $a_i$ and
$b_i$, respectively, are the dimensions of the degree-$i$ spaces in
the chain complexes $\mathcal{A}$ and $\mathcal{B}$.  Thus, the
$Z$-distance of the subsystem code with CSS generators
(\ref{eq:subsystem-product}) may serve as a lower bound for
$d_j(\mathcal{C})$, complimentary to Eq.~(\ref{eq:upper-general}).

Unfortunately, such a projection is not an ideal tool for finding the
minimum distances in the product complex, as the distance may actually
be reduced in some cases.  The examples of such a reduction are based
on Theorem~\ref{ex:upper-distance-bnd-nA} in the previous Section; it
may happen for any finite field.

However, this reduction only concerns the minimum distances in tensor
products of chain complexes after projection to one of the subspaces,
it does not prevent the inequality (\ref{eq:upper-general}) from being
saturated.  We conducted extensive numerical calculations finding
homological distances for products of random $\mathbb{F}_q$-linear
chain complexes with $q\in\{2,3,2^2,5,7,2^3,3^2,11\}$ and space
dimensions of up to 12, and an exhaustive enumeration of products of
binary chain complexes with individual spaces of dimension up to $7$.
Yet we haven't been able to find a single example of a pair of chain
complexes whose product would fail to reach the upper bound
(\ref{eq:upper-general}).  Combined with analytical results for
multiple products of chain complexes involving just two spaces, we
\emph{conjecture} that in general, for any finite field
$F=\mathbb{F}_q$, the homological distances in a tensor product of a
pair of bounded chain complexes of vector spaces over $F$ satisfy the
equality
\begin{equation}
  \label{eq:conjectured-distance}
  d_j(\mathcal{A}\times \mathcal{B})=\min_{i\in\mathbb{Z}}
  d_i(\mathcal{A}) d_{j-i}(\mathcal{B}).  
\end{equation}

\subsection{Upper bound on the distance}

\begin{statement}
  \label{st:upper-bound-complex}
  Consider two $F$-linear chain complexes
  $\mathcal{A}=\mathcal{K}(A_1,\ldots,A_\ell)$ and
  $\mathcal{B}=\mathcal{K}(B_1,\ldots,B_{\ell'})$.  Then, for any
  $i,j\in\mathbb{Z}$, the homological distance of the product complex
  $\mathcal{C}=\mathcal{A}\times \mathcal{B}$ at level $j$
  satisfies the inequality
  \begin{equation}    
    d_{j}(\mathcal{C})\le d_i(\mathcal{A})d_{j-i}(\mathcal{B}).
    \label{eq:upper-bound-one}  
\end{equation}
\label{th:upper-bound-one}
\end{statement}
\begin{proof}
  By definition, the distances $d_i(\mathcal{A})$ and
  $d_{j-i}(\mathcal{B})$ are natural or infinite.  Thus, if one or both
  homology groups are trivial, $k_i(\mathcal{A})=0$ or
  $k_{j-i}(\mathcal{B})=0$ (in which case the corresponding distance is
  infinite), the r.h.s.\ of Eq.~(\ref{eq:upper-bound-one}) equals
  infinity, so that the inequality in question is trivially satisfied.

  Otherwise, with both homology groups non-trivial, consider a pair of
  minimum-weight homologically non-trivial vectors
  $a\in \mathcal{H}_i(\mathcal{A})$ and
  $b\in \mathcal{H}_j(\mathcal{B})$ such that
  $\wgt(a)=d_i(\mathcal{A})$ and $\wgt(b)=d_j(\mathcal{B})$.  Vector
  $a$ is a non-trivial $Z$-like codeword in the stabilizer code
  $\css(A_i,A_{i+1}^T)$; denote $a'$ an $X$-like codeword in the same
  code conjugate to $a$, that is, $a' \cdot a=1$.  In other words,
  this vector is a co-cycle in $\tilde{\cal A}_i$. [In fact, $a'$ is a
  member of the co-homology group $H_i(\tilde{\cal A})$, but this is
  not needed for the proof.]  Similarly, denote $b'$ an $X$-like
  codeword in the code $\css(B_{j-i},B_{j-i+1}^T)$ conjugate to $b$, a
  co-cycle in $\tilde{\cal B}_j$.  Construct $c\in \mathcal{C}_{j}$
  by assigning non-zero value $c_{i,j-i}=a\otimes b$ in the subspace
  $\mathcal{A}_i\otimes \mathcal{B}_{j-i}$, and zero in all other
  subspaces at level $j$.  Clearly,
  $\wgt(c)=d_i(\mathcal{A}) d_{j-i}(\mathcal{B})$; to prove the upper
  bound (\ref{eq:upper-bound-one}) we just need to show that
  $c\not\simeq0$.  To this end, consider a vector $c'$ constructed
  similarly to $c$ but from vectors $a'$ and $b'$; it is easy to check
  that $c\cdot c'=1$.  In addition, this vector is a co-cycle in
  $\tilde{\cal C}_{j}$, i.e., $ c'C_{j+1} =0$, where the matrix is
  a boundary operator in the product complex $\mathcal{C}$,
  cf.~Eq.~(\ref{eq:tp-boundary}).  Any vector equivalent to $c$ has
  the form $c+x (C_{j+1})^T$, for some $x\in\mathcal{C}_{j+1}$.
  However, such a combination is never zero, as can be verified by
  taking a dot product with $c'$.
\end{proof}

The upper bound (\ref{eq:upper-general}) immediately follows from
Statement \ref{th:upper-bound-one} by minimizing over $i$.

\subsection{Lower bounds on the distance}

To make the map with the product codes in Sec.~\ref{sec:product-codes}
evident, we start by writing out the block form of a matrix in the
product complex $\mathcal{C}=\mathcal{A}\otimes \mathcal{B}$, where
the spaces $\mathcal{A}_i$ and $\mathcal{B}_j$ have dimensions $a_i$
and $b_j$, respectively:
\begin{widetext}
\begin{equation}
  \label{eq:Cj}
  C_j=\left(
         \begin{array}[c]{c|c|c|c|c}
           A_j\otimes I(b_0)&(-1)^{j-1} I(a_{j-1})\otimes B_1& &
           \\           \hline
                            &A_{j-1}\otimes
                              I(b_{1})&(-1)^{j-2}I(a_{j-2})\otimes B_2
                                                               &
           \\ \hline
                            & & \ddots & \ddots \\ \hline
           & & & A_1\otimes I(b_{j-1})& I(a_0)\otimes B_{j}
         \end{array}\right).
     \end{equation}
     For ease of mapping of the homology group
$H_j(\mathcal{C})$ to the CSS stabilizer code with generators
$H_X=C_j$ and $H_Z=C_{j+1}^T$, we also write the latter matrix explicitly
\begin{equation} \label{eq:CjT}
C_{j+1}^T=\left(
  \begin{array}[c]{c|c|c|c|c}
    A^T_{j+1}\otimes I({b_0})& & &  \\ \hline
    (-1)^jI(a_j)\otimes B^T_1& A^T_{j}\otimes I(b_1) & &
    \\ \hline
                             & (-1)^{j-1} I(a_{j-1})\otimes B^T_2  
                             & A^T_{j-1}\otimes I(b_2) & \\ \hline
                             & & \ddots & \ddots \\ \hline 
                             & & & -I(a_1)\otimes B_{j-1}^T & A_1^T\otimes I(b_j)
    \\ \hline
&    & & & I(a_0)\otimes B_{j+1}^T
             \end{array}
                                          \right).
\end{equation}

Clearly, in general, the generator matrices $H_X=C_j$ and
$H_Z=C_{j+1}^T$ have $j+1$ column blocks, with each block row and
block column incident on no more than two non-zero blocks.  Our
strategy is to construct bounds on the distance of these codes using
Lemmas \ref{th:punctured} and \ref{th:shortened}.  Notice that for
decomposition along the block boundaries, the condition in Lemma
\ref{th:punctured} can be verified by explicitly constructing bases of
the product chain and product co-chain complexes, and using the
K\"unneth formula to make sure that no vectors are lost.

First, let us construct the codes $\mathcal{Q}^{(i,j-i)}$,
$i\in \mathbb{Z}$, each projected into a single subspace
$\mathcal{A}_i\otimes \mathcal{B}_{j-i}$ as in
Lemma~\ref{th:punctured}.  The corresponding lower bound on the
homological distance at the level $j$ of the product complex
$\mathcal{C}$ reads
\begin{equation}
  d_j(\mathcal{C})\ge \min_{i\in\mathbb{Z}}d_Z\left(\mathcal{Q}^{(i,j-i)}\right).
  \label{eq:block-decomposition}
\end{equation}
Denote $I\equiv I_{i}^j$ the index set corresponding to the subspace
$\mathcal{A}_i\otimes \mathcal{B}_{j-i}$ in $\mathcal{C}_j$.  The
punctured matrix $G_Z[I]$ is obtained by selecting the appropriate
column block in the matrix (\ref{eq:CjT}).  When expressed in terms of
the two small stabilizer codes $\mathcal{Q}_A=\css(A_i,A_{i+1}^T)$ and
$\mathcal{Q}_B=\css(B_{j-i},B_{j-i+1}^T)$ associated with the homology
groups $H_i(\mathcal{A})$ and $H_{j-i}(\mathcal{B})$, respectively,
the resulting matrix has exactly the form of the gauge generator
matrix $G_Z$ in Eq.~(\ref{eq:gen-proj-subs}).  To construct the
matching shortened matrix $(H_X)_I$, notice that only two row blocks
in $C_j$ give non-zero contribution,
$$
C_j[I_{i+1}^j\cup I_{i}^j\cup I_{i-1}^j]=\left(
  \begin{array}[c]{c|c|c}
    (-1)^{i+1} I(a_{i+1})\otimes B_{j-i-1}  \\ \hline
    A_{i+1}\otimes I(b_{j-i-1})
    &(-1)^{i} I(a_{i})\otimes B_{j-i} & \\ \hline
    &A_{i}\otimes I(b_{j-i})&(-1)^{i-1} I(a_{i-1})\otimes B_{j-i+1} \\ \hline
    & & A_{i-1}\otimes I(b_{j-i+1})
  \end{array}\right).
$$
\end{widetext}

The shortening to the middle column block, $I_i^j$, is achieved with the
help of row operations equivalent to left multiplication of the second
row block by $A^*\otimes I(b_{j-i-1})$ and of the third row block by
$I(a_{i-1})\otimes B^*$, where 
\begin{equation} \label{eq:dual-matrices-AB}
  A^*
  =\left(\begin{array}[c]{c}A_i\\ L_X^A
         \end{array}
       \right),\quad
       B^* =\left(\begin{array}[c]{c}B_{j-i}\\ L_X^B
                          \end{array}\right)
\end{equation}
are the largest-rank matrices with rows orthogonal to the columns of
$A_{j+1}$ and $B_{j-i+1}$, respectively.  Here and below, we denote
$L_X^A$, $L_Z^A$ and $L_X^B$, $L_Z^B$ the canonical logical generator
matrices (\ref{eq:cw-orthogonality}) of the same stabilizer codes,
$\mathcal{Q}_A$ and $\mathcal{Q}_B$.  As a result of the
multiplication, we obtain the shortened matrix $(H_X)_I$ in the exact
form of the stabilizer generator matrix $H_X$ in
Eq.~(\ref{eq:stab-proj}), again, when expressed in terms of the
matrices associated with the codes $\mathcal{Q}_A$ and
$\mathcal{Q}_B$.  According to Lemma \ref{th:subs-code-H}, the
corresponding stabilizer code $\css\biglb((H_X)_I,H_Z[I]\bigrb)$ has
exactly the same $Z$-distance as the subsystem code
$\css\biglb(H_X[I],H_Z[I]\bigrb)$ obtained by puncturing both matrices
$H_X={C}_j$ and $H_Z={C}_{j+1}^T$ to the single subspace
$\mathcal{A}_i\otimes \mathcal{B}_{j-i}$.

With the help of the upper bound (\ref{eq:symmetric-distance}) and the
loose lower bound (\ref{eq:loose-lower-bound}), we obtain
\begin{statement}
  \label{th:block-code-distance}
  The $Z$-distance $d_Z\equiv d_Z(\mathcal{Q}^{(i,j-i)})$ of the
  $F$-linear CSS code $\mathcal{Q}^{(i,j-i)}$ obtained by
  $Z$-puncturing the CSS code corresponding to homology group
  $H_j(\mathcal{A}\otimes \mathcal{B})$ to the subspace
  $\mathcal{A}_i\otimes \mathcal{B}_{j-i}$ satisfies the bounds
  \begin{equation}
    \max\biglb( d_i(\mathcal{A}),d_{j-i}(\mathcal{B})\bigrb)\le
    d_Z\le
    d_i(\mathcal{A})d_{j-i}(\mathcal{B}).\label{eq:block-code-distance}
  \end{equation}

\end{statement}

Since $d_0(\mathcal{A})$ and $d_0(\mathcal{B})$ are restricted to be
either zero or infinity, this gives exact values for the distance in
two special cases:
\begin{eqnarray}
  \label{eq:dQ-j-j}
  d_Z(\mathcal{Q}^{(j,0)})&=& d_j(\mathcal{A})d_{0}(\mathcal{B}),\\
  \label{eq:dQ-j-0}
  d_Z(\mathcal{Q}^{(0,j)})&=&   d_0(\mathcal{A})d_{j}(\mathcal{B}).    
\end{eqnarray}
In addition, the structure of the homologically non-trivial vectors is
somewhat clarified by the following restricted result:
\begin{statement}
  \label{th:two-step-bnd}
  Consider a vector $c\in \mathcal{C}_j$ at level $j$ in the product
  chain complex $\mathcal{C}=\mathcal{A}\times \mathcal{B}$, and
  assume that for some $i\le j$, $c$ has a non-zero weight in
  ${H}_{i}(\mathcal{A})\otimes H_{j-i}(\mathcal{B})$, while the
  components of $c$ are zero in spaces
  $\mathcal{A}_{i'}\otimes \mathcal{B}_{j-i'}$ with $i'<i$.  Then
  $\wgt(c)\ge d_{i}(\mathcal{A}) d_{j-i}(\mathcal{B})$.
\end{statement}
\begin{proof}
  The vector is a $Z$-like codeword in the CSS code with generator
  matrices (\ref{eq:Cj}) and (\ref{eq:CjT}).  The condition can be
  used to construct a $Z$-shortened code, with all blocks to the right
  of the block $\mathcal{A}_i\otimes \mathcal{B}_{j-i}$ removed as in
  Lemma \ref{th:shortened}.  This amounts to dropping all column
  blocks of $C_j$ and $C_{j+1}^T$ to the right of the $(j-i+1)$\,th
  block-column which corresponds to the subspace
  $\mathcal{A}_i\otimes \mathcal{B}_{j-i}$, and multiplication of the
  last block-row that remains non-zero in $C_{j+1}^T$ by
  $(A^*)^T\otimes I(b_{j-i})$, where $A^*$ is given by
  Eq.~(\ref{eq:dual-matrices-AB}).  After a subsequent application of
  a $Z$-puncture, so that all block columns to the left of the block
  $\mathcal{A}_i\otimes \mathcal{B}_{j-i}$ are removed as in Lemma
  \ref{th:punctured}, we obtain exactly the concatenated-stabilizer
  code in Theorem \ref{th:concat-stab-code}, constructed from
  $\mathcal{Q}_A=\css(A_i,A_{i+1}^T)$ and
  $\mathcal{Q}_B=\css(B_{j-i},B_{j-i+1}^T)$.  The $Z$-distance of this
  code is $d_Z=d_Z^Ad_Z^B=d_i(\mathcal{A})d_{j-i}(\mathcal{B})$.
  Moreover, by assumption, vector $c$ punctured to the space
  $\mathcal{A}_i\otimes \mathcal{B}_{j-i}$ is non-trivial in the
  product code, which guarantees $\wgt(c)\ge d_Z$.
\end{proof}

Clearly, the same lower bound also applies for vectors with zero
weight in all spaces $\mathcal{A}_{i'}\otimes \mathcal{B}_{j-i'}$ with
$i'>j$.  In addition, the condition of Statement \ref{th:two-step-bnd}
is automatically satisfied when $B_{j-i}$ is the last non-trivial
matrix in the complex $\mathcal{B}$, i.e., $j-i=\ell'$, see Statement
\ref{th:upper-bound-one}.  In this case, again, the upper bound in
Eq.~(\ref{eq:block-code-distance}) is saturated,
$d_Z(\mathcal{Q}^{(i,\ell')})=d_i(\mathcal{A}) d_{\ell'}(\mathcal{B})$.
The same is true also when $A_i$ is the last non-trivial boundary
operator in the complex $\mathcal{A}$, $i=\ell$; we have
$d_Z(\mathcal{Q}^{(\ell,j)})=d_\ell(\mathcal{A}) d_{j}(\mathcal{B})$.

The special cases in Statements \ref{th:block-code-distance} and
\ref{th:two-step-bnd} combine to give exact distances in the case
where one of the complexes in the product contains just one
non-trivial boundary operator.  This gives an extension of the main
result in Ref.~\onlinecite{Zeng-Pryadko-2018} to $F$-linear chain
complexes:

\begin{theorem}
  Consider a tensor product
  $\mathcal{C}=\mathcal{A}\times \mathcal{B}$ of two $F$-linear chain
  complexes, where one of the complexes contains just two non-trivial
  spaces, e.g., $\mathcal{A}=\mathcal{K}(A_1,\ldots,A_\ell)$ and
  $\mathcal{B}=\mathcal{K}(B_1)$.  Then, for any $j\in\mathbb{Z}$,
  the homological distance at level $j$ of the product complex
  $\mathcal{C}=\mathcal{A}\times \mathcal{B}$ is
  \begin{equation}    
    d_{j}(\mathcal{C})=\min_{i\in\mathbb{Z}} d_i(\mathcal{A})d_{j-i}(\mathcal{B}).
    \label{eq:exact-distance}  
\end{equation}
  \label{th:exact-distance}
\end{theorem}

In Ref.~\onlinecite{Zeng-Pryadko-2018}, we conjectured that in the
binary case, $q=2$, the identity (\ref{eq:exact-distance}) be
applicable to products of arbitrary bounded complexes.  The conjecture
was based on extensive numerical simulations of products of
length-three binary complexes corresponding to pairs of
randomly-generated CSS codes.

In addition, here we have conducted numerical simulations of product
chain complexes based on pairs of random $\mathbb{F}_q$-linear
stabilizer codes, with all CSS generators of full-row-rank, so that in
the corresponding chain complexes only the homology groups
$H_1(\mathcal{A})$ and $H_1(\mathcal{B})$ be non-trivial.  For each
$q\in\{2,3,2^2,5,7,2^3,3^2,11\}$, we generated some $2\times 10^4$
such code pairs of length $3\le a_1\le b_1\le 11$, and calculated the
homological distances $d_2(\mathcal{C})$ and $d_2(\tilde{\cal C})$ of
the corresponding (co)chain product complexes using a version of the
covering set
algorithm\cite{Prange-1962,Chua-Yang-1988,Dumer-Kovalev-Pryadko-IEEE-2017}.
Not a single instance was found where the inequality
(\ref{eq:upper-bound-one}) would not be saturated.

Notice that our search went over a tiny fraction of all code pairs, in
particular, since the number of codes (matrices) scales exponentially
with the number of entries, i.e., super-exponentially with the matrix
size.  To ensure that we did not miss any instances, we also
enumerated all pairs of non-trivial binary CSS codes of size $n\le 7$,
and constructed tensor products of the corresponding chain complexes.
Eq.~\ref{eq:exact-distance} was satisfied for all of these.


Based on these numerical results, combined with the analytical result
in Theorem \ref{th:exact-distance} and the results for multiple
products of $1$-complexes, see Sec.~\ref{sec:applications}, we propose
\begin{conjecture}
  \label{conj:product-dist-general}
  The homological distances $d_j(\mathcal{A}\times \mathcal{B})$ in a
  product of any pair of bounded chain complexes of vector spaces over
  a finite field is given by Eq.~(\ref{eq:exact-distance}).
\end{conjecture}

Of course, one should be aware that, even when highly suggestive,
numerical evidence cannot substitute a proof.  A recent example is the
Hedetniemi conjecture about the chromatic number in a tensor product
of graphs\cite{Hedetniemi-1966,Hedetniemi-thesis-1966}.  The
conjecture held up for over half a century; a counterexample was only
recently discovered by Yaroslav Shitov in a beautiful 2019
paper\cite{Shitov-2019,Kalai-2019}.  Significantly, the smallest
graphs known so far to provide a counterexample to Hedetniemi's
conjecture have over $10^4$ vertices\cite{Zhu-2020}.

\subsection{Applications in quantum error correction}
\label{sec:applications}

In classical error correction it is usually safe to assume a channel
model, where errors may happen during transmission but not during
encoding/decoding.  In comparison, when a quantum error-correcting
code (QECC) is used, errors may happen at any step; to measure a
syndrome one has to perform a complex set of elementary quantum
unitaries, gates, which may result not only in additional data errors
but also syndrome measurement errors.  Measurement errors become more
likely with operators of large weight, as the measurement circuit has
to be constructed from elementary quantum gates which typically
can operate at most on two qudits at a time.

As a result, fault-tolerant (FT) operation requires quantum codes
where all (or most) stabilizer generators have small weights.  These
are analogous to classical low-density parity-check (LDPC) codes.

Here we consider tensor products of several $F$-linear $1$-complexes,
chain complexes with just two non-trivial spaces.  Basic parameters
such as space dimensions, row and column weights, or homological
distances do not depend on the order of the terms in the product.
Further, if the matrices used to construct $1$-complexes are
$(\upsilon,\omega)$-sparse, that is, their column and row weights do
not exceed $\upsilon$ and $\omega$, respectively, the matrices in the
resulting $m$-chain product complex are $(m\upsilon,m\omega)$-sparse.
In particular, when ${\cal K}={\cal K}(R)$ is a $1$-complex associated
with a circulant check matrix of the repetition code,
$\mathcal{K}^{\times D}$ recovers all the $D$-dimensional toric codes.

First, consider an $r\times c$ full-row-rank $q$-ary matrix $P$ with
$r<c$, and assume that the $F$-linear code $\mathcal{C}_{P}^\perp$ has
distance $\delta$.  The $1$-complex $\mathcal{K}\equiv \mathcal{K}(P)$
has two non-trivial spaces of dimensions $r$ and $c$; the
corresponding homology groups have ranks $0$, $\kappa$ and the
distances $\infty$, $\delta$.  The $1$-complex
$\tilde{K}\equiv\mathcal{K}({P}^T)$ generated by the transposed matrix
has equivalent spaces taken in the opposite order, with the same
homology group ranks, but the distances are now $1$ and $\infty$,
respectively.  It is easy to see that in any chain complex constructed
as tensor products of $\mathcal{K}$ and/or $\tilde{\mathcal{K}}$,
there is going to be only one homology group with a non-zero rank.
Since order of the products is not important, we will write these as
powers.  For $(a+b)$-complex
$\mathcal{K}^{(a,b)}\equiv\mathcal{K}^{\times a}\times
\tilde{\mathcal{K}}^{\times b}$, the only non-trivial homology group
is $H_a(\mathcal{K}^{(a,b)})$, acting in the space of dimension
$$
n_a(\mathcal{K}^{(a,b)})=\sum_{i=0}^a c^{2i} r^{a+b-2i}{a\choose
i}{b\choose i}<(r+c)^{a+b},
$$
it has rank $\kappa^{a+b}$ and distance $\delta^a$.  The corresponding
quantum CSS code has the minimum distance $\min(\delta^a,\delta^b)$,
and its stabilizer generators have weights not exceeding
$(a+b)\max(\upsilon,\omega)$.

Good weight-limited classical LDPC codes with asymptotically finite
rates $\kappa/c$ and finite relative distances $\delta/c$ can be
obtained from ensembles of large random
matrices\cite{Gallager-1962,Gallager-book-1963,Litsyn-Shevelev-2002,
  Richardson-Shokrollahi-Amin-Urbanke-2001,Frolov-Zyablov-2011,Frolov-2015}.
Any of these can be used in the present construction.  Then, for any
pair $(a,b)$ of natural numbers, we can generate weight-limited
quantum LDPC codes with finite rates and the distances $d_X=\delta^a$,
$d_Z=\delta^b$ whose product scales linearly with the code length.
The quantum hypergraph-product codes are a special case of this
construction with $a=b=1$.

More generally, take arbitrary $r_i\times c_i$ matrices $P_i$,
$i=1,2,\ldots$ with elements from $F\equiv \mathbb{F}_q$.  Let
$F$-linear codes with parity check matrices $P_i$ and $P_i^T$,
respectively, have parameters $[c_i,\kappa_i,\delta_i]$ and
$[r_i,\tilde{\kappa}_i,\tilde{\delta}_i]$, where the distance is
assumed infinite whenever the corresponding code is trivial,
$\kappa=0$.  Then, for a product of $m$ such $1$-complexes, the space
dimensions and ranks of the homology groups following from the
K\"unneth formula can be written in terms of the generating
polynomials
\begin{eqnarray*} n^{(m)}(x)&\equiv &n_0^{(m)}+xn_1^{(m)}+\ldots x^m
n_m^{(m)}\\&=&\prod\nolimits_{j=1}^m(r_j+x c_j),\\ k^{(m)}(x)&\equiv
&k_0^{(m)}+xk_1^{(m)}+\ldots x^m
k_m^{(m)}\\ &=&\prod\nolimits_{j=1}^m(\tilde{\kappa}_j+x \kappa_j).
\end{eqnarray*}
The homological distance $d_j^{(m)}$ can be seen as the minimum over
the products of distances corresponding to those terms that give
non-zero contributions to $k_j^{(m)}$, with the substitution
$\kappa_j\to \delta_j$, $0\neq\tilde{\kappa}_j\to 1$.

It is easy to check that none of the higher-dimensional quantum
hypergraph-product codes discussed here have parameters that are
better than for regular QHP codes ($m=2$) originally constructed by
Tillich and Z{\'e}mor\cite{Tillich-Zemor-2009}.  In addition, the row-
and column-weights of the corresponding matrices tend to get bigger
with increasing $m$.  The advantage of higher-dimension QHP codes, or,
more generally, codes from $m$-chain complexes with $m\ge4$, is that
the rows of matrices $G_X^{(a)}=\mathcal{K}_{a}$,
$G_Z^{(a)}=\mathcal{K}_{a+1}^T$ satisfy a large number of linear
relations resulting from the orthogonality with the matrices $K_{a-1}$
and $K_{a+2}$, respectively.  These can be used to correct syndrome
measurement errors.  Even though the resulting syndrome codes do not
have large distances (with a finite probability some errors remain),
the use of such codes in repeated measurement setting could simplify
the decoding and/or improve the decoding success probability in the
case of adversarial noise\cite{Campbell-2018}.  Such improvements with
stochastic noise have been demonstrated numerically in the case of
$D=4$ toric codes in
Ref.~\onlinecite{Breuckmann-Duivenvoorden-Michels-Terhal-2017}.

\section{Extensions}
\label{sec:extensions}

Throughout this work, we concentrated on the Hamming distance.  A
simple, and yet offering a range of possible applications, extension
of Theorem \ref{th:concat-stab-code}, Statement
\ref{st:upper-bound-complex}, and Theorem \ref{th:exact-distance} can
be given by using \emph{weighted distances}, defined for a vector
$c\in F^n$ in terms of the norm
\begin{equation}
\wgt_W(c)\equiv \sum_{i:c[i]\neq0} W_i,\label{eq:weighted-wgt}
\end{equation}
where $W\equiv (W_1,W_2,\ldots,W_n)$ is a vector of positive weights
$W_i>0$, $i\le n$.  For the corresponding proofs to work, the only
requirement is that the weights $W^{C_{i,j}}$ in each space
$\mathcal{C}_{ij}\equiv \mathcal{A}_i\otimes\mathcal{B}_j$ used to
form the product complex $\mathcal{C}=\mathcal{A}\times \mathcal{B}$
be related to the weights $W^{A_i}$ and $W^{B_j}$ in the original
complexes, namely, $W^{C_{i,j}}=W^{A_i}\otimes W^{B_j}$.  Indeed, all
the proofs are based either on Eq.~(\ref{eq:log-proj}), or a
projection inequality as in Eq.~(\ref{eq:projection-decomposition});
both arguments are readily modified to account for weighted
norm~(\ref{eq:weighted-wgt}).

In particular, this implies an extension to extremal length $L_1$
(systole) and higher-dimensional analogs $L_j$, $j>1$, representing
minimal structures with non-trivial homology on a given
manifold\cite{*[{See, e.g., }] [{}] Pacini-2019}.  Indeed, in the
simplest case, the edge ($j=1$), plaquette ($j=2$), etc.\ weights
associated with a given tessellation can be chosen as the
corresponding Euclidean length, area, etc.  Then the weighted norm
(\ref{eq:weighted-wgt}) gives the corresponding measure of the
elements in the structure, and the homological distance---the
corresponding minimum, going over to $L_j$ in the continuum limit.  We
assume the manifolds be sufficiently smooth so that the corresponding
limits exist\cite{*[{See, e.g.,\ }] [{}] Pierpont-1906}.

Second, an extension of some of the bounds to chain complexes of
$K$-modules, modules over a commuting ring $K$, is possible if $K$ is
a principal ideal domain (PID).  Here we only consider the ring
$K=\mathbb{Z}_q$ of modular integers, and assume torsion-free case,
i.e., with all Smith normal form invariants of all matrices either
zero or one.  In this case one
gets\cite{Jiang-Kovalev-Zeng-Pryadko-unpublished-2020}
$d_Z^C\ge d_Z^Ad_Z^B$ for the stabilizer-product code in Theorem
\ref{th:concat-stab-code}.  Further, the lower bound in
Theorem~(\ref{th:loose-lower-bound}) remains intact, while
Eq.~(\ref{eq:exact-distance}) also becomes an inequality,
$d_j(\mathcal{C})\ge \min_{i\in\mathbb{Z}}d_i(\mathcal{A})
d_{j-i}(\mathcal{B})$.

\medskip\noindent\textbf{Acknowledgments:} We are grateful to Benjamin
Audoux, Daniel Gottesman, and Markus Grassl for helpful conversations.
This work was supported in part by the NSF Division of Physics via
grant No.\ 1820939.

\bibliography{lpp,qc_all,more_qc,ldpc,linalg,percol,sg,spin,analiz}

\begin{thebibliography}{84}%
\makeatletter
\providecommand \@ifxundefined [1]{%
 \@ifx{#1\undefined}
}%
\providecommand \@ifnum [1]{%
 \ifnum #1\expandafter \@firstoftwo
 \else \expandafter \@secondoftwo
 \fi
}%
\providecommand \@ifx [1]{%
 \ifx #1\expandafter \@firstoftwo
 \else \expandafter \@secondoftwo
 \fi
}%
\providecommand \natexlab [1]{#1}%
\providecommand \enquote  [1]{``#1''}%
\providecommand \bibnamefont  [1]{#1}%
\providecommand \bibfnamefont [1]{#1}%
\providecommand \citenamefont [1]{#1}%
\providecommand \href@noop [0]{\@secondoftwo}%
\providecommand \href [0]{\begingroup \@sanitize@url \@href}%
\providecommand \@href[1]{\@@startlink{#1}\@@href}%
\providecommand \@@href[1]{\endgroup#1\@@endlink}%
\providecommand \@sanitize@url [0]{\catcode `\\12\catcode `\$12\catcode
  `\&12\catcode `\#12\catcode `\^12\catcode `\_12\catcode `\%12\relax}%
\providecommand \@@startlink[1]{}%
\providecommand \@@endlink[0]{}%
\providecommand \url  [0]{\begingroup\@sanitize@url \@url }%
\providecommand \@url [1]{\endgroup\@href {#1}{\urlprefix }}%
\providecommand \urlprefix  [0]{URL }%
\providecommand \Eprint [0]{\href }%
\providecommand \doibase [0]{https://doi.org/}%
\providecommand \selectlanguage [0]{\@gobble}%
\providecommand \bibinfo  [0]{\@secondoftwo}%
\providecommand \bibfield  [0]{\@secondoftwo}%
\providecommand \translation [1]{[#1]}%
\providecommand \BibitemOpen [0]{}%
\providecommand \bibitemStop [0]{}%
\providecommand \bibitemNoStop [0]{.\EOS\space}%
\providecommand \EOS [0]{\spacefactor3000\relax}%
\providecommand \BibitemShut  [1]{\csname bibitem#1\endcsname}%
\let\auto@bib@innerbib\@empty
\bibitem [{\citenamefont {Kaczynski}\ \emph {et~al.}(2003)\citenamefont
  {Kaczynski}, \citenamefont {Mischaikow},\ and\ \citenamefont
  {Mrozek}}]{Kaczynski-Mischaikow-Mrozek-2003}%
  \BibitemOpen
  \bibfield  {author} {\bibinfo {author} {\bibfnamefont {T.}~\bibnamefont
  {Kaczynski}}, \bibinfo {author} {\bibfnamefont {K.}~\bibnamefont
  {Mischaikow}},\ and\ \bibinfo {author} {\bibfnamefont {M.}~\bibnamefont
  {Mrozek}},\ }\bibfield  {title} {\bibinfo {title} {Computing homology},\
  }\href@noop {} {\bibfield  {journal} {\bibinfo  {journal} {Homology, Homotopy
  and Applications}\ }\textbf {\bibinfo {volume} {5}},\ \bibinfo {pages} {233}
  (\bibinfo {year} {2003})}\BibitemShut {NoStop}%
\bibitem [{\citenamefont {Kaczynski}\ \emph {et~al.}(2004)\citenamefont
  {Kaczynski}, \citenamefont {Mischaikow},\ and\ \citenamefont
  {Mrozek}}]{Kaczynski-Mischaikow-Mrozek-book-2004}%
  \BibitemOpen
  \bibfield  {author} {\bibinfo {author} {\bibfnamefont {T.}~\bibnamefont
  {Kaczynski}}, \bibinfo {author} {\bibfnamefont {K.}~\bibnamefont
  {Mischaikow}},\ and\ \bibinfo {author} {\bibfnamefont {M.}~\bibnamefont
  {Mrozek}},\ }\href {https://doi.org/10.1007/b97315} {\emph {\bibinfo {title}
  {Computational homology}}},\ \bibinfo {series} {Applied Mathematical
  Sciences}, Vol.\ \bibinfo {volume} {157}\ (\bibinfo  {publisher}
  {Springer-Verlag},\ \bibinfo {address} {New York},\ \bibinfo {year}
  {2004})\BibitemShut {NoStop}%
\bibitem [{\citenamefont {Gonzalez-Lorenzo}\ \emph {et~al.}(2016)\citenamefont
  {Gonzalez-Lorenzo}, \citenamefont {Bac}, \citenamefont {Mari},\ and\
  \citenamefont {Real}}]{GonzalezLorenzo-Bac-Mari-Real-2016}%
  \BibitemOpen
  \bibfield  {author} {\bibinfo {author} {\bibfnamefont {A.}~\bibnamefont
  {Gonzalez-Lorenzo}}, \bibinfo {author} {\bibfnamefont {A.}~\bibnamefont
  {Bac}}, \bibinfo {author} {\bibfnamefont {J.-L.}\ \bibnamefont {Mari}},\ and\
  \bibinfo {author} {\bibfnamefont {P.}~\bibnamefont {Real}},\ }\bibfield
  {title} {\bibinfo {title} {{Two Measures for the Homology Groups of Binary
  Volumes}},\ }in\ \href {https://doi.org/10.1007/978-3-319-32360-2\_12} {\emph
  {\bibinfo {booktitle} {{19th {IAPR} International Conference on Discrete
  Geometry for Computer Imagery ({DGCI} 2016)}}}},\ \bibinfo {series} {Discrete
  Geometry for Computer Imagery}, Vol.\ \bibinfo {volume} {9647},\ \bibinfo
  {editor} {edited by\ \bibinfo {editor} {\bibnamefont {Springer}}}\ (\bibinfo
  {address} {Nantes, France},\ \bibinfo {year} {2016})\ pp.\ \bibinfo {pages}
  {154--165}\BibitemShut {NoStop}%
\bibitem [{\citenamefont {Audoux}\ and\ \citenamefont
  {Couvreur}(2015)}]{Audoux-Couvreur-2017}%
  \BibitemOpen
  \bibfield  {author} {\bibinfo {author} {\bibfnamefont {B.}~\bibnamefont
  {Audoux}}\ and\ \bibinfo {author} {\bibfnamefont {A.}~\bibnamefont
  {Couvreur}},\ }\bibfield  {title} {\bibinfo {title} {On tensor products of
  {CSS} codes},\ }\Eprint {https://arxiv.org/abs/arXiv:1512.07081}
  {arXiv:1512.07081}  (\bibinfo {year} {2015}),\ \bibinfo {note} {to appear in
  Ann. Inst. Henri Poincar\'e D}\BibitemShut {NoStop}%
\bibitem [{\citenamefont {Dey}\ \emph {et~al.}(2018)\citenamefont {Dey},
  \citenamefont {Li},\ and\ \citenamefont {Wang}}]{Dey-Li-Wang-2018}%
  \BibitemOpen
  \bibfield  {author} {\bibinfo {author} {\bibfnamefont {T.~K.}\ \bibnamefont
  {Dey}}, \bibinfo {author} {\bibfnamefont {T.}~\bibnamefont {Li}},\ and\
  \bibinfo {author} {\bibfnamefont {Y.}~\bibnamefont {Wang}},\ }\bibfield
  {title} {\bibinfo {title} {Efficient algorithms for computing a minimal
  homology basis},\ }in\ \href {https://doi.org/10.1007/978-3-319-77404-6_28}
  {\emph {\bibinfo {booktitle} {LATIN 2018: Theoretical Informatics}}},\
  \bibinfo {editor} {edited by\ \bibinfo {editor} {\bibfnamefont {M.~A.}\
  \bibnamefont {Bender}}, \bibinfo {editor} {\bibfnamefont {M.}~\bibnamefont
  {Farach-Colton}},\ and\ \bibinfo {editor} {\bibfnamefont {M.~A.}\
  \bibnamefont {Mosteiro}}}\ (\bibinfo  {publisher} {Springer International
  Publishing},\ \bibinfo {address} {Cham},\ \bibinfo {year} {2018})\ pp.\
  \bibinfo {pages} {376--398}\BibitemShut {NoStop}%
\bibitem [{\citenamefont {Bombin}\ and\ \citenamefont
  {Martin-Delgado}(2007)}]{Bombin-MartinDelgado-2007}%
  \BibitemOpen
  \bibfield  {author} {\bibinfo {author} {\bibfnamefont {H.}~\bibnamefont
  {Bombin}}\ and\ \bibinfo {author} {\bibfnamefont {M.~A.}\ \bibnamefont
  {Martin-Delgado}},\ }\bibfield  {title} {\bibinfo {title} {Homological error
  correction: {C}lassical and quantum codes},\ }\href
  {https://doi.org/10.1063/1.2731356} {\bibfield  {journal} {\bibinfo
  {journal} {Journal of Mathematical Physics}\ }\textbf {\bibinfo {volume}
  {48}},\ \bibinfo {eid} {052105} (\bibinfo {year} {2007})}\BibitemShut
  {NoStop}%
\bibitem [{\citenamefont {Bullock}\ and\ \citenamefont
  {Brennen}(2007)}]{Bullock-Brennen-2007}%
  \BibitemOpen
  \bibfield  {author} {\bibinfo {author} {\bibfnamefont {S.~S.}\ \bibnamefont
  {Bullock}}\ and\ \bibinfo {author} {\bibfnamefont {G.~K.}\ \bibnamefont
  {Brennen}},\ }\bibfield  {title} {\bibinfo {title} {Qudit surface codes and
  gauge theory with finite cyclic groups},\ }\href
  {http://stacks.iop.org/1751-8121/40/i=13/a=013} {\bibfield  {journal}
  {\bibinfo  {journal} {Journal of Physics A: Mathematical and Theoretical}\
  }\textbf {\bibinfo {volume} {40}},\ \bibinfo {pages} {3481} (\bibinfo {year}
  {2007})}\BibitemShut {NoStop}%
\bibitem [{\citenamefont {Calderbank}\ and\ \citenamefont
  {Shor}(1996)}]{Calderbank-Shor-1996}%
  \BibitemOpen
  \bibfield  {author} {\bibinfo {author} {\bibfnamefont {A.~R.}\ \bibnamefont
  {Calderbank}}\ and\ \bibinfo {author} {\bibfnamefont {P.~W.}\ \bibnamefont
  {Shor}},\ }\bibfield  {title} {\bibinfo {title} {Good quantum
  error-correcting codes exist},\ }\href
  {https://doi.org/10.1103/PhysRevA.54.1098} {\bibfield  {journal} {\bibinfo
  {journal} {Phys. Rev. A}\ }\textbf {\bibinfo {volume} {54}},\ \bibinfo
  {pages} {1098} (\bibinfo {year} {1996})}\BibitemShut {NoStop}%
\bibitem [{\citenamefont {Steane}(1996)}]{Steane-1996}%
  \BibitemOpen
  \bibfield  {author} {\bibinfo {author} {\bibfnamefont {A.~M.}\ \bibnamefont
  {Steane}},\ }\bibfield  {title} {\bibinfo {title} {Simple quantum
  error-correcting codes},\ }\href {http://dx.doi.org/10.1103/PhysRevA.54.4741}
  {\bibfield  {journal} {\bibinfo  {journal} {Phys. Rev. A}\ }\textbf {\bibinfo
  {volume} {54}},\ \bibinfo {pages} {4741} (\bibinfo {year}
  {1996})}\BibitemShut {NoStop}%
\bibitem [{\citenamefont {Bravyi}\ and\ \citenamefont
  {Kitaev}(1998)}]{Bravyi-Kitaev-1998}%
  \BibitemOpen
  \bibfield  {author} {\bibinfo {author} {\bibfnamefont {S.~B.}\ \bibnamefont
  {Bravyi}}\ and\ \bibinfo {author} {\bibfnamefont {A.~Y.}\ \bibnamefont
  {Kitaev}},\ }\bibfield  {title} {\bibinfo {title} {Quantum codes on a lattice
  with boundary},\ }\Eprint {https://arxiv.org/abs/quant-ph/9811052}
  {quant-ph/9811052}  (\bibinfo {year} {1998}),\ \bibinfo {note}
  {unpublished}\BibitemShut {NoStop}%
\bibitem [{\citenamefont {Freedman}\ and\ \citenamefont
  {Meyer}(2001)}]{Freedman-Meyer-1998}%
  \BibitemOpen
  \bibfield  {author} {\bibinfo {author} {\bibfnamefont {M.~H.}\ \bibnamefont
  {Freedman}}\ and\ \bibinfo {author} {\bibfnamefont {D.~A.}\ \bibnamefont
  {Meyer}},\ }\bibfield  {title} {\bibinfo {title} {Projective plane and planar
  quantum codes},\ }\href {https://doi.org/10.1007/s102080010013} {\bibfield
  {journal} {\bibinfo  {journal} {Foundations of Computational Mathematics}\
  }\textbf {\bibinfo {volume} {1}},\ \bibinfo {pages} {325} (\bibinfo {year}
  {2001})},\ \Eprint {https://arxiv.org/abs/quant-ph/9810055}
  {quant-ph/9810055} \BibitemShut {NoStop}%
\bibitem [{\citenamefont {Dennis}\ \emph {et~al.}(2002)\citenamefont {Dennis},
  \citenamefont {Kitaev}, \citenamefont {Landahl},\ and\ \citenamefont
  {Preskill}}]{Dennis-Kitaev-Landahl-Preskill-2002}%
  \BibitemOpen
  \bibfield  {author} {\bibinfo {author} {\bibfnamefont {E.}~\bibnamefont
  {Dennis}}, \bibinfo {author} {\bibfnamefont {A.}~\bibnamefont {Kitaev}},
  \bibinfo {author} {\bibfnamefont {A.}~\bibnamefont {Landahl}},\ and\ \bibinfo
  {author} {\bibfnamefont {J.}~\bibnamefont {Preskill}},\ }\bibfield  {title}
  {\bibinfo {title} {Topological quantum memory},\ }\href
  {https://doi.org/10.1063/1.1499754} {\bibfield  {journal} {\bibinfo
  {journal} {J. Math. Phys.}\ }\textbf {\bibinfo {volume} {43}},\ \bibinfo
  {pages} {4452} (\bibinfo {year} {2002})}\BibitemShut {NoStop}%
\bibitem [{\citenamefont {Castelnovo}\ and\ \citenamefont
  {Chamon}(2008)}]{Castelnovo-Chamon-2008}%
  \BibitemOpen
  \bibfield  {author} {\bibinfo {author} {\bibfnamefont {C.}~\bibnamefont
  {Castelnovo}}\ and\ \bibinfo {author} {\bibfnamefont {C.}~\bibnamefont
  {Chamon}},\ }\bibfield  {title} {\bibinfo {title} {Topological order in a
  three-dimensional toric code at finite temperature},\ }\href
  {https://doi.org/10.1103/PhysRevB.78.155120} {\bibfield  {journal} {\bibinfo
  {journal} {Phys. Rev. B}\ }\textbf {\bibinfo {volume} {78}},\ \bibinfo
  {pages} {155120} (\bibinfo {year} {2008})}\BibitemShut {NoStop}%
\bibitem [{\citenamefont {Maz{\'{a}\v{c}}}\ and\ \citenamefont
  {Hamma}(2012)}]{Mazac-Hamma-2012}%
  \BibitemOpen
  \bibfield  {author} {\bibinfo {author} {\bibfnamefont {D.}~\bibnamefont
  {Maz{\'{a}\v{c}}}}\ and\ \bibinfo {author} {\bibfnamefont {A.}~\bibnamefont
  {Hamma}},\ }\bibfield  {title} {\bibinfo {title} {Topological order,
  entanglement, and quantum memory at finite temperature},\ }\href
  {https://doi.org/10.1016/j.aop.2012.05.004} {\bibfield  {journal} {\bibinfo
  {journal} {Annals of Physics}\ }\textbf {\bibinfo {volume} {327}},\ \bibinfo
  {pages} {2096 } (\bibinfo {year} {2012})}\BibitemShut {NoStop}%
\bibitem [{\citenamefont {Bombin}\ \emph {et~al.}(2013)\citenamefont {Bombin},
  \citenamefont {Chhajlany}, \citenamefont {Horodecki},\ and\ \citenamefont
  {Martin-Delgado}}]{Bombin-Chhajlany-Horodecki-MartinDelgado-2013}%
  \BibitemOpen
  \bibfield  {author} {\bibinfo {author} {\bibfnamefont {H.}~\bibnamefont
  {Bombin}}, \bibinfo {author} {\bibfnamefont {R.~W.}\ \bibnamefont
  {Chhajlany}}, \bibinfo {author} {\bibfnamefont {M.}~\bibnamefont
  {Horodecki}},\ and\ \bibinfo {author} {\bibfnamefont {M.~A.}\ \bibnamefont
  {Martin-Delgado}},\ }\bibfield  {title} {\bibinfo {title} {Self-correcting
  quantum computers},\ }\href {http://stacks.iop.org/1367-2630/15/i=5/a=055023}
  {\bibfield  {journal} {\bibinfo  {journal} {New Journal of Physics}\ }\textbf
  {\bibinfo {volume} {15}},\ \bibinfo {pages} {055023} (\bibinfo {year}
  {2013})}\BibitemShut {NoStop}%
\bibitem [{\citenamefont {Kitaev}(2003)}]{kitaev-anyons}%
  \BibitemOpen
  \bibfield  {author} {\bibinfo {author} {\bibfnamefont {A.~Y.}\ \bibnamefont
  {Kitaev}},\ }\bibfield  {title} {\bibinfo {title} {Fault-tolerant quantum
  computation by anyons},\ }\href {http://arxiv.org/abs/quant-ph/9707021}
  {\bibfield  {journal} {\bibinfo  {journal} {Ann. Phys.}\ }\textbf {\bibinfo
  {volume} {303}},\ \bibinfo {pages} {2} (\bibinfo {year} {2003})}\BibitemShut
  {NoStop}%
\bibitem [{\citenamefont {Bravyi}\ and\ \citenamefont
  {Terhal}(2009)}]{Bravyi-Terhal-2009}%
  \BibitemOpen
  \bibfield  {author} {\bibinfo {author} {\bibfnamefont {S.}~\bibnamefont
  {Bravyi}}\ and\ \bibinfo {author} {\bibfnamefont {B.}~\bibnamefont
  {Terhal}},\ }\bibfield  {title} {\bibinfo {title} {A no-go theorem for a
  two-dimensional self-correcting quantum memory based on stabilizer codes},\
  }\href {http://stacks.iop.org/1367-2630/11/i=4/a=043029} {\bibfield
  {journal} {\bibinfo  {journal} {New Journal of Physics}\ }\textbf {\bibinfo
  {volume} {11}},\ \bibinfo {pages} {043029} (\bibinfo {year}
  {2009})}\BibitemShut {NoStop}%
\bibitem [{\citenamefont {Bravyi}\ \emph {et~al.}(2010)\citenamefont {Bravyi},
  \citenamefont {Poulin},\ and\ \citenamefont
  {Terhal}}]{Bravyi-Poulin-Terhal-2010}%
  \BibitemOpen
  \bibfield  {author} {\bibinfo {author} {\bibfnamefont {S.}~\bibnamefont
  {Bravyi}}, \bibinfo {author} {\bibfnamefont {D.}~\bibnamefont {Poulin}},\
  and\ \bibinfo {author} {\bibfnamefont {B.}~\bibnamefont {Terhal}},\
  }\bibfield  {title} {\bibinfo {title} {Tradeoffs for reliable quantum
  information storage in {2D} systems},\ }\href
  {https://doi.org/10.1103/PhysRevLett.104.050503} {\bibfield  {journal}
  {\bibinfo  {journal} {Phys. Rev. Lett.}\ }\textbf {\bibinfo {volume} {104}},\
  \bibinfo {pages} {050503} (\bibinfo {year} {2010})},\ \Eprint
  {https://arxiv.org/abs/0909.5200} {0909.5200} \BibitemShut {NoStop}%
\bibitem [{\citenamefont {Delfosse}(2013)}]{Delfosse-2013}%
  \BibitemOpen
  \bibfield  {author} {\bibinfo {author} {\bibfnamefont {N.}~\bibnamefont
  {Delfosse}},\ }\bibfield  {title} {\bibinfo {title} {Tradeoffs for reliable
  quantum information storage in surface codes and color codes},\ }in\ \href
  {https://doi.org/10.1109/ISIT.2013.6620360} {\emph {\bibinfo {booktitle}
  {Information Theory Proceedings (ISIT), 2013 IEEE International Symposium
  on}}}\ (\bibinfo {organization} {IEEE},\ \bibinfo {year} {2013})\ pp.\
  \bibinfo {pages} {917--921}\BibitemShut {NoStop}%
\bibitem [{\citenamefont {Flammia}\ \emph {et~al.}(2017)\citenamefont
  {Flammia}, \citenamefont {Haah}, \citenamefont {Kastoryano},\ and\
  \citenamefont {Kim}}]{Flammia-Haah-Kastoryano-Kim-2017}%
  \BibitemOpen
  \bibfield  {author} {\bibinfo {author} {\bibfnamefont {S.~T.}\ \bibnamefont
  {Flammia}}, \bibinfo {author} {\bibfnamefont {J.}~\bibnamefont {Haah}},
  \bibinfo {author} {\bibfnamefont {M.~J.}\ \bibnamefont {Kastoryano}},\ and\
  \bibinfo {author} {\bibfnamefont {I.~H.}\ \bibnamefont {Kim}},\ }\bibfield
  {title} {\bibinfo {title} {Limits on the storage of quantum information in a
  volume of space},\ }\href {https://doi.org/10.22331/q-2017-04-25-4}
  {\bibfield  {journal} {\bibinfo  {journal} {{Quantum}}\ }\textbf {\bibinfo
  {volume} {1}},\ \bibinfo {pages} {4} (\bibinfo {year} {2017})},\ \Eprint
  {https://arxiv.org/abs/1610.06169} {1610.06169} \BibitemShut {NoStop}%
\bibitem [{\citenamefont {Kovalev}\ and\ \citenamefont
  {Pryadko}(2013{\natexlab{a}})}]{Kovalev-Pryadko-FT-2013}%
  \BibitemOpen
  \bibfield  {author} {\bibinfo {author} {\bibfnamefont {A.~A.}\ \bibnamefont
  {Kovalev}}\ and\ \bibinfo {author} {\bibfnamefont {L.~P.}\ \bibnamefont
  {Pryadko}},\ }\bibfield  {title} {\bibinfo {title} {Fault tolerance of
  quantum low-density parity check codes with sublinear distance scaling},\
  }\href {https://doi.org/10.1103/PhysRevA.87.020304} {\bibfield  {journal}
  {\bibinfo  {journal} {Phys. Rev. A}\ }\textbf {\bibinfo {volume} {87}},\
  \bibinfo {pages} {020304(R)} (\bibinfo {year}
  {2013}{\natexlab{a}})}\BibitemShut {NoStop}%
\bibitem [{\citenamefont {Dumer}\ \emph {et~al.}(2015)\citenamefont {Dumer},
  \citenamefont {Kovalev},\ and\ \citenamefont
  {Pryadko}}]{Dumer-Kovalev-Pryadko-bnd-2015}%
  \BibitemOpen
  \bibfield  {author} {\bibinfo {author} {\bibfnamefont {I.}~\bibnamefont
  {Dumer}}, \bibinfo {author} {\bibfnamefont {A.~A.}\ \bibnamefont {Kovalev}},\
  and\ \bibinfo {author} {\bibfnamefont {L.~P.}\ \bibnamefont {Pryadko}},\
  }\bibfield  {title} {\bibinfo {title} {Thresholds for correcting errors,
  erasures, and faulty syndrome measurements in degenerate quantum codes},\
  }\href {https://doi.org/10.1103/PhysRevLett.115.050502} {\bibfield  {journal}
  {\bibinfo  {journal} {Phys. Rev. Lett.}\ }\textbf {\bibinfo {volume} {115}},\
  \bibinfo {pages} {050502} (\bibinfo {year} {2015})},\ \Eprint
  {https://arxiv.org/abs/1412.6172} {1412.6172} \BibitemShut {NoStop}%
\bibitem [{\citenamefont
  {Gottesman}(2014{\natexlab{a}})}]{Gottesman-overhead-2014}%
  \BibitemOpen
  \bibfield  {author} {\bibinfo {author} {\bibfnamefont {D.}~\bibnamefont
  {Gottesman}},\ }\bibfield  {title} {\bibinfo {title} {Fault-tolerant quantum
  computation with constant overhead},\ }\href
  {http://www.rintonpress.com/journals/qicabstracts/qicabstracts14-1516.html}
  {\bibfield  {journal} {\bibinfo  {journal} {Quant. Information and
  Computation}\ }\textbf {\bibinfo {volume} {14}},\ \bibinfo {pages} {1338}
  (\bibinfo {year} {2014}{\natexlab{a}})},\ \Eprint
  {https://arxiv.org/abs/1310.2984} {1310.2984} \BibitemShut {NoStop}%
\bibitem [{\citenamefont {Gallager}(1962)}]{Gallager-1962}%
  \BibitemOpen
  \bibfield  {author} {\bibinfo {author} {\bibfnamefont {R.}~\bibnamefont
  {Gallager}},\ }\bibfield  {title} {\bibinfo {title} {Low-density parity-check
  codes},\ }\href {https://doi.org/10.1109/TIT.1962.1057683} {\bibfield
  {journal} {\bibinfo  {journal} {IRE Trans. Inf. Theory}\ }\textbf {\bibinfo
  {volume} {8}},\ \bibinfo {pages} {21} (\bibinfo {year} {1962})}\BibitemShut
  {NoStop}%
\bibitem [{\citenamefont {Gallager}(1963)}]{Gallager-book-1963}%
  \BibitemOpen
  \bibfield  {author} {\bibinfo {author} {\bibfnamefont {R.~G.}\ \bibnamefont
  {Gallager}},\ }\href {https://doi.org/doi=10.1.1.147.683} {\emph {\bibinfo
  {title} {Low-Density Parity-Check Codes}}}\ (\bibinfo  {publisher} {M.I.T.
  Press},\ \bibinfo {address} {Cambridge, Mass.},\ \bibinfo {year}
  {1963})\BibitemShut {NoStop}%
\bibitem [{\citenamefont {Litsyn}\ and\ \citenamefont
  {Shevelev}(2002)}]{Litsyn-Shevelev-2002}%
  \BibitemOpen
  \bibfield  {author} {\bibinfo {author} {\bibfnamefont {S.}~\bibnamefont
  {Litsyn}}\ and\ \bibinfo {author} {\bibfnamefont {V.}~\bibnamefont
  {Shevelev}},\ }\bibfield  {title} {\bibinfo {title} {On ensembles of
  low-density parity-check codes: asymptotic distance distributions},\ }\href
  {https://doi.org/10.1109/18.992777} {\bibfield  {journal} {\bibinfo
  {journal} {IEEE Trans. Inf. Theory}\ }\textbf {\bibinfo {volume} {48}},\
  \bibinfo {pages} {887} (\bibinfo {year} {2002})}\BibitemShut {NoStop}%
\bibitem [{\citenamefont {Richardson}\ \emph {et~al.}(2001)\citenamefont
  {Richardson}, \citenamefont {Shokrollahi},\ and\ \citenamefont
  {Urbanke}}]{Richardson-Shokrollahi-Amin-Urbanke-2001}%
  \BibitemOpen
  \bibfield  {author} {\bibinfo {author} {\bibfnamefont {T.~J.}\ \bibnamefont
  {Richardson}}, \bibinfo {author} {\bibfnamefont {M.~A.}\ \bibnamefont
  {Shokrollahi}},\ and\ \bibinfo {author} {\bibfnamefont {R.~L.}\ \bibnamefont
  {Urbanke}},\ }\bibfield  {title} {\bibinfo {title} {Design of
  capacity-approaching irregular low-density parity-check codes},\ }\href
  {https://doi.org/10.1109/18.910578} {\bibfield  {journal} {\bibinfo
  {journal} {Information Theory, IEEE Transactions on}\ }\textbf {\bibinfo
  {volume} {47}},\ \bibinfo {pages} {619} (\bibinfo {year} {2001})}\BibitemShut
  {NoStop}%
\bibitem [{\citenamefont {Z{\'e}mor}(2009)}]{Zemor-2009}%
  \BibitemOpen
  \bibfield  {author} {\bibinfo {author} {\bibfnamefont {G.}~\bibnamefont
  {Z{\'e}mor}},\ }\bibfield  {title} {\bibinfo {title} {On {C}ayley graphs,
  surface codes, and the limits of homological coding for quantum error
  correction},\ }in\ \href {https://doi.org/10.1007/978-3-642-01877-0_21}
  {\emph {\bibinfo {booktitle} {Coding and Cryptology: Second International
  Workshop, {IWCC} 2009, Proc.}}},\ \bibinfo {editor} {edited by\ \bibinfo
  {editor} {\bibfnamefont {Y.~M.}\ \bibnamefont {Chee}}, \bibinfo {editor}
  {\bibfnamefont {C.}~\bibnamefont {Li}}, \bibinfo {editor} {\bibfnamefont
  {S.}~\bibnamefont {Ling}}, \bibinfo {editor} {\bibfnamefont {H.}~\bibnamefont
  {Wang}},\ and\ \bibinfo {editor} {\bibfnamefont {C.}~\bibnamefont {Xing}}}\
  (\bibinfo  {publisher} {Springer},\ \bibinfo {address} {Berlin, Heidelberg},\
  \bibinfo {year} {2009})\ pp.\ \bibinfo {pages} {259--273}\BibitemShut
  {NoStop}%
\bibitem [{\citenamefont {Delfosse}\ and\ \citenamefont
  {Z\'{e}mor}(2010)}]{Delfosse-Zemor-2010}%
  \BibitemOpen
  \bibfield  {author} {\bibinfo {author} {\bibfnamefont {N.}~\bibnamefont
  {Delfosse}}\ and\ \bibinfo {author} {\bibfnamefont {G.}~\bibnamefont
  {Z\'{e}mor}},\ }\bibfield  {title} {\bibinfo {title} {Quantum
  erasure-correcting codes and percolation on regular tilings of the hyperbolic
  plane},\ }in\ \href {https://doi.org/10.1109/CIG.2010.5592863} {\emph
  {\bibinfo {booktitle} {Information Theory Workshop (ITW), 2010 IEEE}}}\
  (\bibinfo {year} {2010})\ pp.\ \bibinfo {pages} {1--5}\BibitemShut {NoStop}%
\bibitem [{\citenamefont {Breuckmann}\ and\ \citenamefont
  {Terhal}(2016)}]{Breuckmann-Terhal-2015}%
  \BibitemOpen
  \bibfield  {author} {\bibinfo {author} {\bibfnamefont {N.~P.}\ \bibnamefont
  {Breuckmann}}\ and\ \bibinfo {author} {\bibfnamefont {B.~M.}\ \bibnamefont
  {Terhal}},\ }\bibfield  {title} {\bibinfo {title} {Constructions and noise
  threshold of hyperbolic surface codes},\ }\href
  {https://doi.org/10.1109/TIT.2016.2555700} {\bibfield  {journal} {\bibinfo
  {journal} {IEEE Trans. on Inf. Th.}\ }\textbf {\bibinfo {volume} {62}},\
  \bibinfo {pages} {3731} (\bibinfo {year} {2016})},\ \Eprint
  {https://arxiv.org/abs/1506.04029} {1506.04029} \BibitemShut {NoStop}%
\bibitem [{\citenamefont {Breuckmann}\ \emph
  {et~al.}(2017{\natexlab{a}})\citenamefont {Breuckmann}, \citenamefont
  {Vuillot}, \citenamefont {Campbell}, \citenamefont {Krishna},\ and\
  \citenamefont {Terhal}}]{Breuckmann-Vuillot-Campbell-Krishna-Terhal-2017}%
  \BibitemOpen
  \bibfield  {author} {\bibinfo {author} {\bibfnamefont {N.~P.}\ \bibnamefont
  {Breuckmann}}, \bibinfo {author} {\bibfnamefont {C.}~\bibnamefont {Vuillot}},
  \bibinfo {author} {\bibfnamefont {E.}~\bibnamefont {Campbell}}, \bibinfo
  {author} {\bibfnamefont {A.}~\bibnamefont {Krishna}},\ and\ \bibinfo {author}
  {\bibfnamefont {B.~M.}\ \bibnamefont {Terhal}},\ }\bibfield  {title}
  {\bibinfo {title} {Hyperbolic and semi-hyperbolic surface codes for quantum
  storage},\ }\href {http://stacks.iop.org/2058-9565/2/i=3/a=035007} {\bibfield
   {journal} {\bibinfo  {journal} {Quantum Science and Technology}\ }\textbf
  {\bibinfo {volume} {2}},\ \bibinfo {pages} {035007} (\bibinfo {year}
  {2017}{\natexlab{a}})}\BibitemShut {NoStop}%
\bibitem [{\citenamefont {Guth}\ and\ \citenamefont
  {Lubotzky}(2014)}]{Guth-Lubotzky-2014}%
  \BibitemOpen
  \bibfield  {author} {\bibinfo {author} {\bibfnamefont {L.}~\bibnamefont
  {Guth}}\ and\ \bibinfo {author} {\bibfnamefont {A.}~\bibnamefont
  {Lubotzky}},\ }\bibfield  {title} {\bibinfo {title} {Quantum error correcting
  codes and 4-dimensional arithmetic hyperbolic manifolds},\ }\href
  {https://doi.org/http://dx.doi.org/10.1063/1.4891487} {\bibfield  {journal}
  {\bibinfo  {journal} {Journal of Mathematical Physics}\ }\textbf {\bibinfo
  {volume} {55}},\ \bibinfo {pages} {082202} (\bibinfo {year} {2014})},\
  \Eprint {https://arxiv.org/abs/arXiv:1310.5555} {arXiv:1310.5555}
  \BibitemShut {NoStop}%
\bibitem [{\citenamefont {Tillich}\ and\ \citenamefont
  {Z{\'e}mor}(2009)}]{Tillich-Zemor-2009}%
  \BibitemOpen
  \bibfield  {author} {\bibinfo {author} {\bibfnamefont {J.-P.}\ \bibnamefont
  {Tillich}}\ and\ \bibinfo {author} {\bibfnamefont {G.}~\bibnamefont
  {Z{\'e}mor}},\ }\bibfield  {title} {\bibinfo {title} {Quantum {LDPC} codes
  with positive rate and minimum distance proportional to {$\sqrt{n}$}},\ }in\
  \href {https://doi.org/10.1109/ISIT.2009.5205648} {\emph {\bibinfo
  {booktitle} {Proc. IEEE Int. Symp. Inf. Theory (ISIT)}}}\ (\bibinfo {year}
  {2009})\ pp.\ \bibinfo {pages} {799--803}\BibitemShut {NoStop}%
\bibitem [{\citenamefont {Tillich}\ and\ \citenamefont
  {Z{\'e}mor}(2014)}]{Tillich-Zemor-2014}%
  \BibitemOpen
  \bibfield  {author} {\bibinfo {author} {\bibfnamefont {J.-P.}\ \bibnamefont
  {Tillich}}\ and\ \bibinfo {author} {\bibfnamefont {G.}~\bibnamefont
  {Z{\'e}mor}},\ }\bibfield  {title} {\bibinfo {title} {Quantum {LDPC} codes
  with positive rate and minimum distance proportional to the square root of
  the blocklength},\ }\href {https://doi.org/10.1109/TIT.2013.229206}
  {\bibfield  {journal} {\bibinfo  {journal} {IEEE Transactions on Information
  Theory}\ }\textbf {\bibinfo {volume} {60}},\ \bibinfo {pages} {1193}
  (\bibinfo {year} {2014})}\BibitemShut {NoStop}%
\bibitem [{\citenamefont {Kovalev}\ and\ \citenamefont
  {Pryadko}(2013{\natexlab{b}})}]{Kovalev-Pryadko-Hyperbicycle-2013}%
  \BibitemOpen
  \bibfield  {author} {\bibinfo {author} {\bibfnamefont {A.~A.}\ \bibnamefont
  {Kovalev}}\ and\ \bibinfo {author} {\bibfnamefont {L.~P.}\ \bibnamefont
  {Pryadko}},\ }\bibfield  {title} {\bibinfo {title} {Quantum {K}ronecker
  sum-product low-density parity-check codes with finite rate},\ }\href
  {https://doi.org/10.1103/PhysRevA.88.012311} {\bibfield  {journal} {\bibinfo
  {journal} {Phys. Rev. A}\ }\textbf {\bibinfo {volume} {88}},\ \bibinfo
  {pages} {012311} (\bibinfo {year} {2013}{\natexlab{b}})}\BibitemShut
  {NoStop}%
\bibitem [{\citenamefont {Zeng}\ and\ \citenamefont
  {Pryadko}(2019)}]{Zeng-Pryadko-2018}%
  \BibitemOpen
  \bibfield  {author} {\bibinfo {author} {\bibfnamefont {W.}~\bibnamefont
  {Zeng}}\ and\ \bibinfo {author} {\bibfnamefont {L.~P.}\ \bibnamefont
  {Pryadko}},\ }\bibfield  {title} {\bibinfo {title} {Higher-dimensional
  quantum hypergraph-product codes with finite rates},\ }\href
  {https://doi.org/10.1103/PhysRevLett.122.230501} {\bibfield  {journal}
  {\bibinfo  {journal} {Phys. Rev. Lett.}\ }\textbf {\bibinfo {volume} {122}},\
  \bibinfo {pages} {230501} (\bibinfo {year} {2019})},\ \Eprint
  {https://arxiv.org/abs/1810.01519} {1810.01519} \BibitemShut {NoStop}%
\bibitem [{\citenamefont {Couvreur}\ \emph {et~al.}(2012)\citenamefont
  {Couvreur}, \citenamefont {Delfosse},\ and\ \citenamefont
  {Z{\'e}mor}}]{Couvreur-Delfosse-Zemor-2012}%
  \BibitemOpen
  \bibfield  {author} {\bibinfo {author} {\bibfnamefont {A.}~\bibnamefont
  {Couvreur}}, \bibinfo {author} {\bibfnamefont {N.}~\bibnamefont {Delfosse}},\
  and\ \bibinfo {author} {\bibfnamefont {G.}~\bibnamefont {Z{\'e}mor}},\
  }\bibfield  {title} {\bibinfo {title} {A construction of quantum {LDPC} codes
  from {Cayley} graphs},\ }\href {http://arxiv.org/abs/1206.2656} {\bibfield
  {journal} {\bibinfo  {journal} {CoRR}\ }\textbf {\bibinfo {volume}
  {abs/1206.2656}} (\bibinfo {year} {2012})},\ \Eprint
  {https://arxiv.org/abs/arXiv:1206.2656} {arXiv:1206.2656} \BibitemShut
  {NoStop}%
\bibitem [{\citenamefont {Bravyi}\ and\ \citenamefont
  {Hastings}(2013)}]{Bravyi-Hastings-2013}%
  \BibitemOpen
  \bibfield  {author} {\bibinfo {author} {\bibfnamefont {S.}~\bibnamefont
  {Bravyi}}\ and\ \bibinfo {author} {\bibfnamefont {M.~B.}\ \bibnamefont
  {Hastings}},\ }\bibfield  {title} {\bibinfo {title} {Homological product
  codes},\ }\Eprint {https://arxiv.org/abs/arXiv:1311.0885} {arXiv:1311.0885}
  (\bibinfo {year} {2013}),\ \bibinfo {note} {unpublished}\BibitemShut
  {NoStop}%
\bibitem [{\citenamefont {Audoux}(2014)}]{Audoux-2014}%
  \BibitemOpen
  \bibfield  {author} {\bibinfo {author} {\bibfnamefont {B.}~\bibnamefont
  {Audoux}},\ }\bibfield  {title} {\bibinfo {title} {An application of
  {K}hovanov homology to quantum codes},\ }\href
  {https://doi.org/10.4171/AIHPD/6} {\bibfield  {journal} {\bibinfo  {journal}
  {Ann.\ Inst.\ Henri Poincare}\ }\textbf {\bibinfo {volume} {1}},\ \bibinfo
  {pages} {185–223} (\bibinfo {year} {2014})}\BibitemShut {NoStop}%
\bibitem [{\citenamefont {Hastings}(2016{\natexlab{a}})}]{Hastings-codes-2016}%
  \BibitemOpen
  \bibfield  {author} {\bibinfo {author} {\bibfnamefont {M.~B.}\ \bibnamefont
  {Hastings}},\ }\bibfield  {title} {\bibinfo {title} {Quantum codes from
  high-dimensional manifolds},\ }\Eprint
  {https://arxiv.org/abs/arXiv:1608.05089} {arXiv:1608.05089}  (\bibinfo {year}
  {2016}{\natexlab{a}}),\ \bibinfo {note} {unpublished}\BibitemShut {NoStop}%
\bibitem [{\citenamefont
  {Hastings}(2016{\natexlab{b}})}]{Hastings-weight-2016}%
  \BibitemOpen
  \bibfield  {author} {\bibinfo {author} {\bibfnamefont {M.~B.}\ \bibnamefont
  {Hastings}},\ }\bibfield  {title} {\bibinfo {title} {Weight reduction for
  quantum codes},\ }\Eprint {https://arxiv.org/abs/arXiv:1611.03790}
  {arXiv:1611.03790}  (\bibinfo {year} {2016}{\natexlab{b}}),\ \bibinfo {note}
  {unpublished}\BibitemShut {NoStop}%
\bibitem [{\citenamefont {Hastings}\ \emph {et~al.}(2020)\citenamefont
  {Hastings}, \citenamefont {Haah},\ and\ \citenamefont
  {{O'Donnell}}}]{Hastings-Haah-ODonnell-2020}%
  \BibitemOpen
  \bibfield  {author} {\bibinfo {author} {\bibfnamefont {M.~B.}\ \bibnamefont
  {Hastings}}, \bibinfo {author} {\bibfnamefont {J.}~\bibnamefont {Haah}},\
  and\ \bibinfo {author} {\bibfnamefont {R.}~\bibnamefont {{O'Donnell}}},\
  }\bibfield  {title} {\bibinfo {title} {Fiber bundle codes: Breaking the
  {$N^{1/2}\mathop{\rm polylog}(N)$} barrier for quantum {LDPC} codes},\
  }\Eprint {https://arxiv.org/abs/2009.03921} {2009.03921}  (\bibinfo {year}
  {2020}),\ \bibinfo {note} {unpublished}\BibitemShut {NoStop}%
\bibitem [{\citenamefont {K{\"u}nneth}(1923{\natexlab{a}})}]{Kunneth-1923}%
  \BibitemOpen
  \bibfield  {author} {\bibinfo {author} {\bibfnamefont {H.}~\bibnamefont
  {K{\"u}nneth}},\ }\bibfield  {title} {\bibinfo {title} {Ueber die bettische
  zahlen einer produktmannigfaltigkeit},\ }\href@noop {} {\bibfield  {journal}
  {\bibinfo  {journal} {Math. Ann.}\ }\textbf {\bibinfo {volume} {90}},\
  \bibinfo {pages} {65} (\bibinfo {year} {1923}{\natexlab{a}})}\BibitemShut
  {NoStop}%
\bibitem [{\citenamefont {K{\"u}nneth}(1923{\natexlab{b}})}]{Kunneth-1924}%
  \BibitemOpen
  \bibfield  {author} {\bibinfo {author} {\bibfnamefont {H.}~\bibnamefont
  {K{\"u}nneth}},\ }\bibfield  {title} {\bibinfo {title} {Ueber die
  torsionszahlen von produktmannigfaltigkeiten},\ }\href@noop {} {\bibfield
  {journal} {\bibinfo  {journal} {Math. Ann.}\ }\textbf {\bibinfo {volume}
  {91}},\ \bibinfo {pages} {125} (\bibinfo {year}
  {1923}{\natexlab{b}})}\BibitemShut {NoStop}%
\bibitem [{\citenamefont {Wang}\ \emph {et~al.}(2003)\citenamefont {Wang},
  \citenamefont {Harrington},\ and\ \citenamefont
  {Preskill}}]{Wang-Harrington-Preskill-2003}%
  \BibitemOpen
  \bibfield  {author} {\bibinfo {author} {\bibfnamefont {C.}~\bibnamefont
  {Wang}}, \bibinfo {author} {\bibfnamefont {J.}~\bibnamefont {Harrington}},\
  and\ \bibinfo {author} {\bibfnamefont {J.}~\bibnamefont {Preskill}},\
  }\bibfield  {title} {\bibinfo {title} {Confinement-{H}iggs transition in a
  disordered gauge theory and the accuracy threshold for quantum memory},\
  }\href {https://doi.org/http://dx.doi.org/10.1016/S0003-4916(02)00019-2}
  {\bibfield  {journal} {\bibinfo  {journal} {Annals of Physics}\ }\textbf
  {\bibinfo {volume} {303}},\ \bibinfo {pages} {31 } (\bibinfo {year}
  {2003})}\BibitemShut {NoStop}%
\bibitem [{\citenamefont {Andrist}(2012)}]{Andrist-PhD-2012}%
  \BibitemOpen
  \bibfield  {author} {\bibinfo {author} {\bibfnamefont {R.~S.}\ \bibnamefont
  {Andrist}},\ }\emph {\bibinfo {title} {Understanding topological quantum
  error-correction codes using classical spin models}},\ \href
  {https://doi.org/https://doi.org/10.3929/ethz-a-007594594} {Ph.D. thesis},\
  \bibinfo  {school} {ETH Z\"urich}, \bibinfo {address} {Dep. of Physics, ETH
  Z\"urich} (\bibinfo {year} {2012})\BibitemShut {NoStop}%
\bibitem [{\citenamefont {Fujiwara}(2014)}]{Fujiwara-2014}%
  \BibitemOpen
  \bibfield  {author} {\bibinfo {author} {\bibfnamefont {Y.}~\bibnamefont
  {Fujiwara}},\ }\bibfield  {title} {\bibinfo {title} {Ability of stabilizer
  quantum error correction to protect itself from its own imperfection},\
  }\href {https://doi.org/10.1103/PhysRevA.90.062304} {\bibfield  {journal}
  {\bibinfo  {journal} {Phys. Rev. A}\ }\textbf {\bibinfo {volume} {90}},\
  \bibinfo {pages} {062304} (\bibinfo {year} {2014})}\BibitemShut {NoStop}%
\bibitem [{\citenamefont {Ashikhmin}\ \emph {et~al.}(2014)\citenamefont
  {Ashikhmin}, \citenamefont {Lai},\ and\ \citenamefont
  {Brun}}]{Ashikhmin-Lai-Brun-2014}%
  \BibitemOpen
  \bibfield  {author} {\bibinfo {author} {\bibfnamefont {A.}~\bibnamefont
  {Ashikhmin}}, \bibinfo {author} {\bibfnamefont {C.~Y.}\ \bibnamefont {Lai}},\
  and\ \bibinfo {author} {\bibfnamefont {T.~A.}\ \bibnamefont {Brun}},\
  }\bibfield  {title} {\bibinfo {title} {Robust quantum error syndrome
  extraction by classical coding},\ }in\ \href
  {https://doi.org/10.1109/ISIT.2014.6874892} {\emph {\bibinfo {booktitle}
  {2014 {IEEE} International Symposium on Information Theory}}}\ (\bibinfo
  {year} {2014})\ pp.\ \bibinfo {pages} {546--550}\BibitemShut {NoStop}%
\bibitem [{\citenamefont {Ashikhmin}\ \emph {et~al.}(2016)\citenamefont
  {Ashikhmin}, \citenamefont {Lai},\ and\ \citenamefont
  {Brun}}]{Ashikhmin-Lai-Brun-2016}%
  \BibitemOpen
  \bibfield  {author} {\bibinfo {author} {\bibfnamefont {A.}~\bibnamefont
  {Ashikhmin}}, \bibinfo {author} {\bibfnamefont {C.~Y.}\ \bibnamefont {Lai}},\
  and\ \bibinfo {author} {\bibfnamefont {T.~A.}\ \bibnamefont {Brun}},\
  }\bibfield  {title} {\bibinfo {title} {Correction of data and syndrome errors
  by stabilizer codes},\ }in\ \href {https://doi.org/10.1109/ISIT.2016.7541704}
  {\emph {\bibinfo {booktitle} {2016 IEEE International Symposium on
  Information Theory (ISIT)}}}\ (\bibinfo {year} {2016})\ pp.\ \bibinfo {pages}
  {2274--2278},\ \Eprint {https://arxiv.org/abs/arXiv:1602.01545}
  {arXiv:1602.01545} \BibitemShut {NoStop}%
\bibitem [{\citenamefont {Campbell}(2019)}]{Campbell-2018}%
  \BibitemOpen
  \bibfield  {author} {\bibinfo {author} {\bibfnamefont {E.~T.}\ \bibnamefont
  {Campbell}},\ }\bibfield  {title} {\bibinfo {title} {A theory of single-shot
  error correction for adversarial noise},\ }\href
  {https://doi.org/10.1088/2058-9565/aafc8f} {\bibfield  {journal} {\bibinfo
  {journal} {Quantum Science and Technology}\ }\textbf {\bibinfo {volume}
  {4}},\ \bibinfo {pages} {025006} (\bibinfo {year} {2019})},\ \Eprint
  {https://arxiv.org/abs/1805.09271} {1805.09271} \BibitemShut {NoStop}%
\bibitem [{\citenamefont {Ioffe}\ and\ \citenamefont
  {M\'ezard}(2007)}]{Ioffe-Mezard-2007}%
  \BibitemOpen
  \bibfield  {author} {\bibinfo {author} {\bibfnamefont {L.}~\bibnamefont
  {Ioffe}}\ and\ \bibinfo {author} {\bibfnamefont {M.}~\bibnamefont
  {M\'ezard}},\ }\bibfield  {title} {\bibinfo {title} {Asymmetric quantum
  error-correcting codes},\ }\href
  {http://dx.doi.org/10.1103/PhysRevA.75.032345} {\bibfield  {journal}
  {\bibinfo  {journal} {Phys. Rev. A}\ }\textbf {\bibinfo {volume} {75}},\
  \bibinfo {pages} {032345} (\bibinfo {year} {2007})}\BibitemShut {NoStop}%
\bibitem [{\citenamefont {Evans}\ \emph {et~al.}(2007)\citenamefont {Evans},
  \citenamefont {Stephens}, \citenamefont {Cole},\ and\ \citenamefont
  {Hollenberg}}]{Evans-2007}%
  \BibitemOpen
  \bibfield  {author} {\bibinfo {author} {\bibfnamefont {Z.~W.~E.}\
  \bibnamefont {Evans}}, \bibinfo {author} {\bibfnamefont {A.~M.}\ \bibnamefont
  {Stephens}}, \bibinfo {author} {\bibfnamefont {J.~H.}\ \bibnamefont {Cole}},\
  and\ \bibinfo {author} {\bibfnamefont {L.~C.~L.}\ \bibnamefont
  {Hollenberg}},\ }\bibfield  {title} {\bibinfo {title} {Error correction
  optimisation in the presence of x/z asymmetry}} (\bibinfo {year} {2007}),\
  \bibinfo {note} {arXiv:0709.3875}\BibitemShut {NoStop}%
\bibitem [{\citenamefont {Stephens}\ \emph {et~al.}(2008)\citenamefont
  {Stephens}, \citenamefont {Evans}, \citenamefont {Devitt},\ and\
  \citenamefont {Hollenberg}}]{Stephens-2008}%
  \BibitemOpen
  \bibfield  {author} {\bibinfo {author} {\bibfnamefont {A.~M.}\ \bibnamefont
  {Stephens}}, \bibinfo {author} {\bibfnamefont {Z.~W.~E.}\ \bibnamefont
  {Evans}}, \bibinfo {author} {\bibfnamefont {S.~J.}\ \bibnamefont {Devitt}},\
  and\ \bibinfo {author} {\bibfnamefont {L.~C.~L.}\ \bibnamefont
  {Hollenberg}},\ }\bibfield  {title} {\bibinfo {title} {Asymmetric quantum
  error correction via code conversion},\ }\href
  {http://dx.doi.org/10.1103/PhysRevA.77.062335} {\bibfield  {journal}
  {\bibinfo  {journal} {Phys. Rev. A}\ }\textbf {\bibinfo {volume} {77}},\
  \bibinfo {pages} {062335} (\bibinfo {year} {2008})}\BibitemShut {NoStop}%
\bibitem [{\citenamefont {Aliferis}\ and\ \citenamefont
  {Preskill}(2008)}]{Aliferis-Preskill-2008}%
  \BibitemOpen
  \bibfield  {author} {\bibinfo {author} {\bibfnamefont {P.}~\bibnamefont
  {Aliferis}}\ and\ \bibinfo {author} {\bibfnamefont {J.}~\bibnamefont
  {Preskill}},\ }\bibfield  {title} {\bibinfo {title} {Fault-tolerant quantum
  computation against biased noise},\ }\href
  {http://dx.doi.org/10.1103/PhysRevA.78.052331} {\bibfield  {journal}
  {\bibinfo  {journal} {Phys. Rev. A}\ }\textbf {\bibinfo {volume} {78}},\
  \bibinfo {pages} {052331} (\bibinfo {year} {2008})}\BibitemShut {NoStop}%
\bibitem [{\citenamefont {Sarvepalli}\ \emph {et~al.}(2009)\citenamefont
  {Sarvepalli}, \citenamefont {Klappenecker},\ and\ \citenamefont
  {R\"otteler}}]{sarvepalli-2009}%
  \BibitemOpen
  \bibfield  {author} {\bibinfo {author} {\bibfnamefont {P.~K.}\ \bibnamefont
  {Sarvepalli}}, \bibinfo {author} {\bibfnamefont {A.}~\bibnamefont
  {Klappenecker}},\ and\ \bibinfo {author} {\bibfnamefont {M.}~\bibnamefont
  {R\"otteler}},\ }\bibfield  {title} {\bibinfo {title} {Asymmetric quantum
  codes: constructions, bounds and performance},\ }\href
  {http://dx.doi.org/10.1098/rspa.2008.0439} {\bibfield  {journal} {\bibinfo
  {journal} {Proc. R. Soc. A}\ }\textbf {\bibinfo {volume} {465}},\ \bibinfo
  {pages} {1645} (\bibinfo {year} {2009})}\BibitemShut {NoStop}%
\bibitem [{\citenamefont {Tuckett}\ \emph {et~al.}(2018)\citenamefont
  {Tuckett}, \citenamefont {Bartlett},\ and\ \citenamefont
  {Flammia}}]{Tuckett-Bartlett-Flammia-2018}%
  \BibitemOpen
  \bibfield  {author} {\bibinfo {author} {\bibfnamefont {D.~K.}\ \bibnamefont
  {Tuckett}}, \bibinfo {author} {\bibfnamefont {S.~D.}\ \bibnamefont
  {Bartlett}},\ and\ \bibinfo {author} {\bibfnamefont {S.~T.}\ \bibnamefont
  {Flammia}},\ }\bibfield  {title} {\bibinfo {title} {Ultrahigh error threshold
  for surface codes with biased noise},\ }\href
  {https://doi.org/10.1103/PhysRevLett.120.050505} {\bibfield  {journal}
  {\bibinfo  {journal} {Phys. Rev. Lett.}\ }\textbf {\bibinfo {volume} {120}},\
  \bibinfo {pages} {050505} (\bibinfo {year} {2018})}\BibitemShut {NoStop}%
\bibitem [{\citenamefont {Weibel}(1994)}]{Weibel-book-1994}%
  \BibitemOpen
  \bibfield  {author} {\bibinfo {author} {\bibfnamefont {C.~A.}\ \bibnamefont
  {Weibel}},\ }\href {https://doi.org/10.1017/CBO9781139644136} {\emph
  {\bibinfo {title} {An introduction to homological algebra}}},\ \bibinfo
  {series} {Cambridge Studies in Advanced Mathematics}, Vol.~\bibinfo {volume}
  {38}\ (\bibinfo  {publisher} {Cambridge University Press},\ \bibinfo
  {address} {Cambridge},\ \bibinfo {year} {1994})\BibitemShut {NoStop}%
\bibitem [{\citenamefont {MacWilliams}\ and\ \citenamefont
  {Sloane}(1981)}]{MS-book}%
  \BibitemOpen
  \bibfield  {author} {\bibinfo {author} {\bibfnamefont {F.~J.}\ \bibnamefont
  {MacWilliams}}\ and\ \bibinfo {author} {\bibfnamefont {N.~J.~A.}\
  \bibnamefont {Sloane}},\ }\href@noop {} {\emph {\bibinfo {title} {The Theory
  of Error-Correcting Codes}}}\ (\bibinfo  {publisher} {North-Holland},\
  \bibinfo {address} {Amsterdam},\ \bibinfo {year} {1981})\BibitemShut
  {NoStop}%
\bibitem [{\citenamefont {Rains}\ and\ \citenamefont
  {Sloane}(1998)}]{Rains-Sloane-self-dual-codes-1998}%
  \BibitemOpen
  \bibfield  {author} {\bibinfo {author} {\bibfnamefont {E.~M.}\ \bibnamefont
  {Rains}}\ and\ \bibinfo {author} {\bibfnamefont {N.~J.~A.}\ \bibnamefont
  {Sloane}},\ }\bibfield  {title} {\bibinfo {title} {Self-dual codes},\ }in\
  \href {http://www2.research.att.com/~njas/doc/self.pdf} {\emph {\bibinfo
  {booktitle} {Handbook of coding theory}}},\ Vol.~\bibinfo {volume} {I},\
  \bibinfo {editor} {edited by\ \bibinfo {editor} {\bibfnamefont {V.~S.~P.}\
  \bibnamefont {and. W.~C.~Huffman}}}\ (\bibinfo  {publisher} {North-Holland},\
  \bibinfo {address} {Amsterdam},\ \bibinfo {year} {1998})\ pp.\ \bibinfo
  {pages} {177--294}\BibitemShut {NoStop}%
\bibitem [{\citenamefont {Ketkar}\ \emph {et~al.}(2006)\citenamefont {Ketkar},
  \citenamefont {Klappenecker}, \citenamefont {Kumar},\ and\ \citenamefont
  {Sarvepalli}}]{Ketkar-Klappenecker-Kumar-Sarvepalli-2006}%
  \BibitemOpen
  \bibfield  {author} {\bibinfo {author} {\bibfnamefont {A.}~\bibnamefont
  {Ketkar}}, \bibinfo {author} {\bibfnamefont {A.}~\bibnamefont
  {Klappenecker}}, \bibinfo {author} {\bibfnamefont {S.}~\bibnamefont
  {Kumar}},\ and\ \bibinfo {author} {\bibfnamefont {P.~K.}\ \bibnamefont
  {Sarvepalli}},\ }\bibfield  {title} {\bibinfo {title} {Nonbinary stabilizer
  codes over finite fields},\ }\href {https://doi.org/10.1109/TIT.2006.883612}
  {\bibfield  {journal} {\bibinfo  {journal} {{IEEE} Trans. Info. Th.}\
  }\textbf {\bibinfo {volume} {52}},\ \bibinfo {pages} {4892} (\bibinfo {year}
  {2006})},\ \Eprint {https://arxiv.org/abs/arXiv:quant-ph/0508070}
  {arXiv:quant-ph/0508070} \BibitemShut {NoStop}%
\bibitem [{\citenamefont {Fan}\ and\ \citenamefont
  {Zgang}(2017)}]{Fan-Zgang-2017}%
  \BibitemOpen
  \bibfield  {author} {\bibinfo {author} {\bibfnamefont {Y.}~\bibnamefont
  {Fan}}\ and\ \bibinfo {author} {\bibfnamefont {L.}~\bibnamefont {Zgang}},\
  }\bibfield  {title} {\bibinfo {title} {Galois self-dual constacyclic codes},\
  }\href {https://doi.org/10.1007/s10623-016-0282-8} {\bibfield  {journal}
  {\bibinfo  {journal} {Designs, Codes and Cryptography}\ }\textbf {\bibinfo
  {volume} {84}},\ \bibinfo {pages} {473} (\bibinfo {year} {2017})},\ \Eprint
  {https://arxiv.org/abs/arXiv:1512.07347} {arXiv:1512.07347} \BibitemShut
  {NoStop}%
\bibitem [{\citenamefont {Ashikhmin}\ and\ \citenamefont
  {Knill}(2001)}]{Ashikhmin-Knill-2001}%
  \BibitemOpen
  \bibfield  {author} {\bibinfo {author} {\bibfnamefont {A.}~\bibnamefont
  {Ashikhmin}}\ and\ \bibinfo {author} {\bibfnamefont {E.}~\bibnamefont
  {Knill}},\ }\bibfield  {title} {\bibinfo {title} {Nonbinary quantum
  stabilizer codes},\ }\href {https://doi.org/10.1109/18.959288} {\bibfield
  {journal} {\bibinfo  {journal} {IEEE Trans. Info. Th.}\ }\textbf {\bibinfo
  {volume} {47}},\ \bibinfo {pages} {3065} (\bibinfo {year}
  {2001})}\BibitemShut {NoStop}%
\bibitem [{\citenamefont
  {Gottesman}(2014{\natexlab{b}})}]{Gottesman-prime-power-2014}%
  \BibitemOpen
  \bibfield  {author} {\bibinfo {author} {\bibfnamefont {D.}~\bibnamefont
  {Gottesman}},\ }\bibfield  {title} {\bibinfo {title} {Stabilizer codes with
  prime power qudits}} (\bibinfo {year} {2014}{\natexlab{b}}),\ \bibinfo {note}
  {invited talk at QEC 2014 (ETH Zurich)}\BibitemShut {NoStop}%
\bibitem [{\citenamefont {Poulin}(2005)}]{Poulin-subs-2005}%
  \BibitemOpen
  \bibfield  {author} {\bibinfo {author} {\bibfnamefont {D.}~\bibnamefont
  {Poulin}},\ }\bibfield  {title} {\bibinfo {title} {Stabilizer formalism for
  operator quantum error correction},\ }\href
  {https://doi.org/10.1103/PhysRevLett.95.230504} {\bibfield  {journal}
  {\bibinfo  {journal} {Phys. Rev. Lett.}\ }\textbf {\bibinfo {volume} {95}},\
  \bibinfo {pages} {230504} (\bibinfo {year} {2005})}\BibitemShut {NoStop}%
\bibitem [{\citenamefont {Bacon}(2006)}]{Bacon-subs-2006}%
  \BibitemOpen
  \bibfield  {author} {\bibinfo {author} {\bibfnamefont {D.}~\bibnamefont
  {Bacon}},\ }\bibfield  {title} {\bibinfo {title} {Operator quantum
  error-correcting subsystems for self-correcting quantum memories},\ }\href
  {https://doi.org/10.1103/PhysRevA.73.012340} {\bibfield  {journal} {\bibinfo
  {journal} {Phys. Rev. A}\ }\textbf {\bibinfo {volume} {73}},\ \bibinfo
  {pages} {012340} (\bibinfo {year} {2006})}\BibitemShut {NoStop}%
\bibitem [{\citenamefont {Li}\ and\ \citenamefont
  {Yoder}(2020)}]{Li-Yoder-2020}%
  \BibitemOpen
  \bibfield  {author} {\bibinfo {author} {\bibfnamefont {M.}~\bibnamefont
  {Li}}\ and\ \bibinfo {author} {\bibfnamefont {T.~J.}\ \bibnamefont {Yoder}},\
  }\bibfield  {title} {\bibinfo {title} {A numerical study of
  {B}ravyi-{B}acon-{S}hor and subsystem hypergraph product codes},\ }\Eprint
  {https://arxiv.org/abs/arXiv:2002.06257} {arXiv:2002.06257}  (\bibinfo {year}
  {2020}),\ \bibinfo {note} {unpublished}\BibitemShut {NoStop}%
\bibitem [{\citenamefont {Shor}(1996)}]{Shor-FT-1996}%
  \BibitemOpen
  \bibfield  {author} {\bibinfo {author} {\bibfnamefont {P.~W.}\ \bibnamefont
  {Shor}},\ }\bibfield  {title} {\bibinfo {title} {Fault-tolerant quantum
  computation},\ }in\ \href {http://arxiv.org/abs/quant-ph/9605011v2} {\emph
  {\bibinfo {booktitle} {Proc. 37th Ann. Symp. on Fundamentals of Comp.
  Sci.}}},\ \bibinfo {organization} {IEEE}\ (\bibinfo  {publisher} {IEEE Comp.
  Soc. Press},\ \bibinfo {address} {Los Alamitos},\ \bibinfo {year} {1996})\
  pp.\ \bibinfo {pages} {56--65},\ \Eprint
  {https://arxiv.org/abs/quant-ph/9605011} {quant-ph/9605011} \BibitemShut
  {NoStop}%
\bibitem [{\citenamefont {Li}\ \emph {et~al.}(2018)\citenamefont {Li},
  \citenamefont {Miller}, \citenamefont {Newman}, \citenamefont {Wu},\ and\
  \citenamefont {Brown}}]{Li-etal-Brown-2018}%
  \BibitemOpen
  \bibfield  {author} {\bibinfo {author} {\bibfnamefont {M.}~\bibnamefont
  {Li}}, \bibinfo {author} {\bibfnamefont {D.}~\bibnamefont {Miller}}, \bibinfo
  {author} {\bibfnamefont {M.}~\bibnamefont {Newman}}, \bibinfo {author}
  {\bibfnamefont {Y.}~\bibnamefont {Wu}},\ and\ \bibinfo {author}
  {\bibfnamefont {K.~R.}\ \bibnamefont {Brown}},\ }\bibfield  {title} {\bibinfo
  {title} {2-{D} compass codes},\ }\Eprint {https://arxiv.org/abs/1809.01193}
  {1809.01193}  (\bibinfo {year} {2018}),\ \bibinfo {note}
  {unpublished}\BibitemShut {NoStop}%
\bibitem [{\citenamefont {Yoder}(2019)}]{Yoder-2019}%
  \BibitemOpen
  \bibfield  {author} {\bibinfo {author} {\bibfnamefont {T.~J.}\ \bibnamefont
  {Yoder}},\ }\bibfield  {title} {\bibinfo {title} {Optimal quantum subsystem
  codes in 2-dimensions},\ }\Eprint {https://arxiv.org/abs/1901.06319}
  {1901.06319}  (\bibinfo {year} {2019}),\ \bibinfo {note}
  {unpublished}\BibitemShut {NoStop}%
\bibitem [{\citenamefont {Napp}\ and\ \citenamefont
  {Preskill}(2013)}]{Napp-Preskill-2013}%
  \BibitemOpen
  \bibfield  {author} {\bibinfo {author} {\bibfnamefont {J.}~\bibnamefont
  {Napp}}\ and\ \bibinfo {author} {\bibfnamefont {J.}~\bibnamefont
  {Preskill}},\ }\bibfield  {title} {\bibinfo {title} {Optimal {B}acon-{S}hor
  codes},\ }\href@noop {} {\bibfield  {journal} {\bibinfo  {journal} {Quant.
  Inf. \& Comp.}\ }\textbf {\bibinfo {volume} {Vol.13 2013}},\ \bibinfo {pages}
  {0490} (\bibinfo {year} {2013})},\ \Eprint
  {https://arxiv.org/abs/arXiv:1209.0794} {arXiv:1209.0794} \BibitemShut
  {NoStop}%
\bibitem [{\citenamefont {Prange}(1962)}]{Prange-1962}%
  \BibitemOpen
  \bibfield  {author} {\bibinfo {author} {\bibfnamefont {E.}~\bibnamefont
  {Prange}},\ }\bibfield  {title} {\bibinfo {title} {The use of information
  sets in decoding cyclic codes},\ }\href
  {https://doi.org/10.1109/TIT.1962.1057777} {\bibfield  {journal} {\bibinfo
  {journal} {Information Theory, IRE Transactions on}\ }\textbf {\bibinfo
  {volume} {8}},\ \bibinfo {pages} {5} (\bibinfo {year} {1962})}\BibitemShut
  {NoStop}%
\bibitem [{\citenamefont {Chua}\ and\ \citenamefont
  {Yang}(1988)}]{Chua-Yang-1988}%
  \BibitemOpen
  \bibfield  {author} {\bibinfo {author} {\bibfnamefont {L.~O.}\ \bibnamefont
  {Chua}}\ and\ \bibinfo {author} {\bibfnamefont {L.}~\bibnamefont {Yang}},\
  }\bibfield  {title} {\bibinfo {title} {Cellular neural networks: theory},\
  }\href {https://doi.org/10.1109/31.7600} {\bibfield  {journal} {\bibinfo
  {journal} {IEEE Transactions on Circuits and Systems}\ }\textbf {\bibinfo
  {volume} {35}},\ \bibinfo {pages} {1257} (\bibinfo {year}
  {1988})}\BibitemShut {NoStop}%
\bibitem [{\citenamefont {Dumer}\ \emph {et~al.}(2017)\citenamefont {Dumer},
  \citenamefont {Kovalev},\ and\ \citenamefont
  {Pryadko}}]{Dumer-Kovalev-Pryadko-IEEE-2017}%
  \BibitemOpen
  \bibfield  {author} {\bibinfo {author} {\bibfnamefont {I.}~\bibnamefont
  {Dumer}}, \bibinfo {author} {\bibfnamefont {A.~A.}\ \bibnamefont {Kovalev}},\
  and\ \bibinfo {author} {\bibfnamefont {L.~P.}\ \bibnamefont {Pryadko}},\
  }\bibfield  {title} {\bibinfo {title} {Distance verification for classical
  and quantum {LDPC} codes},\ }\href {https://doi.org/10.1109/TIT.2017.2690381}
  {\bibfield  {journal} {\bibinfo  {journal} {IEEE Trans. Inf. Th.}\ }\textbf
  {\bibinfo {volume} {63}},\ \bibinfo {pages} {4675} (\bibinfo {year}
  {2017})}\BibitemShut {NoStop}%
\bibitem [{\citenamefont {Hedetniemi}(1966{\natexlab{a}})}]{Hedetniemi-1966}%
  \BibitemOpen
  \bibfield  {author} {\bibinfo {author} {\bibfnamefont {S.~T.}\ \bibnamefont
  {Hedetniemi}},\ }\href@noop {} {\emph {\bibinfo {title} {Homomorphisms of
  graphs and automata}}},\ \bibinfo {type} {Tech. Rep.}\ \bibinfo {number}
  {03105-44-T}\ (\bibinfo  {institution} {University of Michigan},\ \bibinfo
  {year} {(1966)})\BibitemShut {NoStop}%
\bibitem [{\citenamefont
  {Hedetniemi}(1966{\natexlab{b}})}]{Hedetniemi-thesis-1966}%
  \BibitemOpen
  \bibfield  {author} {\bibinfo {author} {\bibfnamefont {S.~T.}\ \bibnamefont
  {Hedetniemi}},\ }\emph {\bibinfo {title} {Homomorphisms of graphs and
  automata}},\ \href
  {https://www.researchgate.net/publication/35444301_Homomorphisms_of_graphs_and_automata}
  {Ph.D. thesis},\ \bibinfo  {school} {The University of Michigan} (\bibinfo
  {year} {1966}{\natexlab{b}})\BibitemShut {NoStop}%
\bibitem [{\citenamefont {Shitov}(2019)}]{Shitov-2019}%
  \BibitemOpen
  \bibfield  {author} {\bibinfo {author} {\bibfnamefont {Y.}~\bibnamefont
  {Shitov}},\ }\bibfield  {title} {\bibinfo {title} {Counterexamples to
  {H}edetniemi's conjecture},\ }\Eprint
  {https://arxiv.org/abs/arXiv:1905.02167} {arXiv:1905.02167}  (\bibinfo {year}
  {2019}),\ \bibinfo {note} {unpublished}\BibitemShut {NoStop}%
\bibitem [{\citenamefont {Kalai}(2019)}]{Kalai-2019}%
  \BibitemOpen
  \bibfield  {author} {\bibinfo {author} {\bibfnamefont {G.}~\bibnamefont
  {Kalai}},\ }\bibfield  {title} {\bibinfo {title} {A sensation in the morning
  news---{Y}aroslav {S}hitov: {C}ounterexamples to {H}edetniemi's conjecture},\
  }\href
  {https://gilkalai.wordpress.com/2019/05/10/sansation-in-the-morning-news-yaroslav-shitov-counterexamples-to-hedetniemis-conjecture/}
  {\bibfield  {journal} {\bibinfo  {journal} {Combinatorics and More}\ }
  (\bibinfo {year} {2019})}\BibitemShut {NoStop}%
\bibitem [{\citenamefont {Zhu}(2020)}]{Zhu-2020}%
  \BibitemOpen
  \bibfield  {author} {\bibinfo {author} {\bibfnamefont {X.}~\bibnamefont
  {Zhu}},\ }\bibfield  {title} {\bibinfo {title} {Relatively small
  counterexamples to {H}edetniemi's conjecture},\ }\Eprint
  {https://arxiv.org/abs/arXiv:2004.09028} {arXiv:2004.09028}  (\bibinfo {year}
  {2020}),\ \bibinfo {note} {unpublished}\BibitemShut {NoStop}%
\bibitem [{\citenamefont {Frolov}\ and\ \citenamefont
  {Zyablov}(2011)}]{Frolov-Zyablov-2011}%
  \BibitemOpen
  \bibfield  {author} {\bibinfo {author} {\bibfnamefont {A.~A.}\ \bibnamefont
  {Frolov}}\ and\ \bibinfo {author} {\bibfnamefont {V.~V.}\ \bibnamefont
  {Zyablov}},\ }\bibfield  {title} {\bibinfo {title} {Bounds on the minimum
  code distance for nonbinary codes based on bipartite graphs},\ }\href
  {https://doi.org/10.1134/S0032946011040028} {\bibfield  {journal} {\bibinfo
  {journal} {Problems of Information Transmission}\ }\textbf {\bibinfo {volume}
  {47}},\ \bibinfo {pages} {327} (\bibinfo {year} {2011})}\BibitemShut
  {NoStop}%
\bibitem [{\citenamefont {Frolov}(2015)}]{Frolov-2015}%
  \BibitemOpen
  \bibfield  {author} {\bibinfo {author} {\bibfnamefont {A.}~\bibnamefont
  {Frolov}},\ }\bibfield  {title} {\bibinfo {title} {An upper bound on the
  minimum distance of {LDPC} codes over {GF(q)}},\ }in\ \href
  {https://doi.org/10.1109/ISIT.2015.7282984} {\emph {\bibinfo {booktitle}
  {2015 IEEE International Symposium on Information Theory (ISIT)}}}\ (\bibinfo
  {year} {2015})\ pp.\ \bibinfo {pages} {2885--2888}\BibitemShut {NoStop}%
\bibitem [{\citenamefont {Breuckmann}\ \emph
  {et~al.}(2017{\natexlab{b}})\citenamefont {Breuckmann}, \citenamefont
  {Duivenvoorden}, \citenamefont {Michels},\ and\ \citenamefont
  {Terhal}}]{Breuckmann-Duivenvoorden-Michels-Terhal-2017}%
  \BibitemOpen
  \bibfield  {author} {\bibinfo {author} {\bibfnamefont {N.~P.}\ \bibnamefont
  {Breuckmann}}, \bibinfo {author} {\bibfnamefont {K.}~\bibnamefont
  {Duivenvoorden}}, \bibinfo {author} {\bibfnamefont {D.}~\bibnamefont
  {Michels}},\ and\ \bibinfo {author} {\bibfnamefont {B.~M.}\ \bibnamefont
  {Terhal}},\ }\bibfield  {title} {\bibinfo {title} {Local decoders for the
  {2D} and {4D} toric code},\ }\href@noop {} {\bibfield  {journal} {\bibinfo
  {journal} {Quantum Inf. Comput.}\ }\textbf {\bibinfo {volume} {17}},\
  \bibinfo {pages} {0181} (\bibinfo {year} {2017}{\natexlab{b}})},\ \Eprint
  {https://arxiv.org/abs/1609.00510} {1609.00510} \BibitemShut {NoStop}%
\bibitem [{\citenamefont {Pacini}(2019)}]{Pacini-2019}%
  \BibitemOpen
  \bibfield  {author} {\bibinfo {author} {\bibfnamefont {T.}~\bibnamefont
  {Pacini}},\ }\bibfield  {title} {\bibinfo {title} {Extremal length in higher
  dimensions and complex systolic inequalities},\ }\Eprint
  {https://arxiv.org/abs/arXiv:1904.07807} {arXiv:1904.07807}  (\bibinfo {year}
  {2019}),\ \bibinfo {note} {unpublished}\BibitemShut {NoStop}%
\bibitem [{\citenamefont {Pierpont}(1906)}]{Pierpont-1906}%
  \BibitemOpen
  \bibfield  {author} {\bibinfo {author} {\bibfnamefont {J.}~\bibnamefont
  {Pierpont}},\ }\bibfield  {title} {\bibinfo {title} {Area of curved
  surfaces},\ }\href {https://doi.org/10.2307/1986240} {\bibfield  {journal}
  {\bibinfo  {journal} {Transactions of the American Mathematical Society}\
  }\textbf {\bibinfo {volume} {7}},\ \bibinfo {pages} {489} (\bibinfo {year}
  {1906})}\BibitemShut {NoStop}%
\bibitem [{\citenamefont {Jiang}\ \emph {et~al.}(2020)\citenamefont {Jiang},
  \citenamefont {Kovalev}, \citenamefont {Zeng},\ and\ \citenamefont
  {Pryadko}}]{Jiang-Kovalev-Zeng-Pryadko-unpublished-2020}%
  \BibitemOpen
  \bibfield  {author} {\bibinfo {author} {\bibfnamefont {Y.}~\bibnamefont
  {Jiang}}, \bibinfo {author} {\bibfnamefont {A.~A.}\ \bibnamefont {Kovalev}},
  \bibinfo {author} {\bibfnamefont {W.}~\bibnamefont {Zeng}},\ and\ \bibinfo
  {author} {\bibfnamefont {L.~P.}\ \bibnamefont {Pryadko}},\ }\bibfield
  {title} {\bibinfo {title} {Quantum codes over $\mathbb{Z}_q$ and $q$-state
  {P}otts models}} (\bibinfo {year} {2020}),\ \bibinfo {note}
  {unpublished}\BibitemShut {NoStop}%
\end{thebibliography}%

\end{document}